\pdfoutput=1
 
\documentclass[11pt]{article}
\usepackage[numbers]{natbib}
\usepackage[section]{placeins}
\usepackage{xspace}
\usepackage[table,xcdraw]{xcolor}
\definecolor{darkblue}{rgb}{0,0,.5}
\usepackage[colorlinks=true,allcolors=darkblue]{hyperref}
\usepackage{authblk}

\usepackage{algorithm}
\usepackage[noend]{algpseudocode}
\usepackage{bbm}
\usepackage{macros}
\usepackage{tikz}
\usetikzlibrary{matrix}
\usepackage{pgfplots}
\usepackage{overpic}
\usepackage{makecell}
\usepgfplotslibrary{fillbetween}
\usetikzlibrary{patterns}
\usepackage{enumerate}
\usepackage{amssymb}
\usepackage{amsbsy}
\usepackage{amsmath}
\usepackage{amsthm}
\usepackage{amsfonts}
\usepackage{latexsym}
\usepackage{graphicx}
\usepackage{color}
\usepackage{algorithmicx, algpseudocode}
\usepackage{makecell}
\usepackage{cellspace} %
\usepackage{tabularx}

\usepackage{subcaption}
\usepackage{listings}
\newcommand{\R}{\mathbb{R}}

\newcommand{\N}{\mathbb{N}}

\newtheorem{theorem}{Theorem}
\newtheorem{lemma}{Lemma}[section]

\newtheorem{definition}{Definition}[section]

\renewcommand\log{\ln}

\usepackage{fullpage}

\begin{document}

\title{ 
LinkedIn's Audience Engagements API: A Privacy Preserving Data Analytics System at Scale}
\author[1]{Ryan Rogers}
\author[1]{Subbu Subramaniam}
\author[1]{Sean Peng}
\author[1]{David Durfee}
\author[1]{Seunghyun Lee}
\author[1]{Santosh Kumar Kancha}
\author[1]{Shraddha Sahay}
\author[1]{Parvez Ahammad}

\affil[1]{LinkedIn Corporation}

\maketitle

\begin{abstract}
We present a privacy system that leverages differential privacy to protect LinkedIn members' data while also providing audience engagement insights to enable marketing analytics related applications.  We detail the differentially private algorithms and other privacy safeguards used to provide results that can be used with existing real-time data analytics platforms, specifically with the open sourced Pinot system.  Our privacy system provides user-level privacy guarantees.   As part of our privacy system, we include a budget management service that enforces a strict differential privacy budget on the returned results to the analyst.  This budget management service brings together the latest research in differential privacy into a product to maintain utility given a fixed differential privacy budget.  
\end{abstract}


\section{Introduction}
LinkedIn's Audience Engagement API is a platform that enables marketers (analysts) aggregated insights about members' content engagements while ensuring member (user) data is protected.  Consider an advertiser that is selling a cloud solution and wants to create a sponsored post on LinkedIn.  The advertiser might use the Audience Engagement API to do research and find that the target audience engages with GDPR articles.  Hence, the advertiser should write about how their cloud solution adheres to GDPR standards, thus increasing engagement.  By design, the Audience Engagement API is secure, aggregated, and uses state of the art differentially private algorithms to provide rigorous privacy guarantees.  The data honors user privacy settings and only contains information that is approved as outlined by GDPR.  Further, data is purged within 30 days of a user leaving the ecosystem because it has a 30 day retention.  

The primary reason to leverage differential privacy is due to \emph{differencing attacks}, where the difference between two queries reveals an individual's content.  For example, one query can be for the top articles by engagement from CEOs in India and then another query asks for the top articles by engagement from CEOs in India or LinkedIn.  This attack makes aggregation and thresholding approaches insufficient --- two results having counts above a threshold does not mean that their difference cannot uniquely identify an individual.  Rather than limiting the scope of the Audience Engagement API by reducing the ways the dataset can be sliced in an ad hoc way, we instead worked to include differentially private algorithms to prevent such differencing attacks.  To make such a product private, it is apparent that one must add noise, to prevent differencing attacks, and limit the number of accesses to the API, to prevent reconstructing the dataset, despite the noise that is added.  Differential privacy then formalizes these approaches via randomized algorithms and its composition properties.  

To incorporate differential privacy, we carefully balance various resources, including data storage distributed across several servers, real-time query computation, privacy loss quantified by the differential privacy parameters $(\diffp,\delta)$, and accuracy.  We describe here the overall privacy system deployed at LinkedIn that balances these resources to provide a product that surfaces audience engagement insights while putting members first by safeguarding their data.  

Providing scalable, real-time analytics with low latency without differential privacy is challenging enough.  Luckily, we have the open source real-time distributed OLAP datastore, called Apache Pinot (incubating) \cite{PinotSigmod}.  Pinot enables use cases like \emph{Job and Publisher Analytics} and \emph{Who Viewed My Profile}.  In order to develop a differentially private system, we need to think how it can be used in conjunction with a (distributed) OLAP system such as Pinot.  This would enable us to have scalable privacy systems.  Furthermore, we need to implement a budgeting tool into the API so that analysts cannot repeatedly query the dataset thus making noise addition pointless.  Our goal is twofold: 
implement differentially private algorithms that can be used with real-time distributed OLAP systems 
and incorporate a privacy budget management service to restrict the amount of information an analyst can retrieve.   
For the privacy budget management service, we incorporate the latest composition bounds for our particular algorithms \cite{DongDuRo19} to extract more utility subject to a given differential privacy budget.

\subsection{Contributions}
We make several contributions toward making practical privacy systems that leverage differential privacy.  
\begin{itemize}
\item We describe a suite of differentially private algorithms that cover the data analytics tasks for LinkedIn's Audience Engagement API, which provide user-level privacy guarantees.
\item We detail our privacy budget management service that is able to track each analyst's privacy budget over multiple data centers.  Hence, we can ensure the budget is enforced across large scale systems in real-time.
\item We showcase empirical results of our algorithms on LinkedIn's data for various privacy parameters on our deployed system.
\item We provide a discussion about the considerations in our privacy system, in particular how we rationalize certain parameters.  We hope that this discussion will help guide practitioners in how parameters might be set and provide transparency into our deployed system.
\end{itemize}

Although the private algorithms and privacy budget formulas were known in prior work, the main contribution of this work is in combining both algorithms and budget management into a system that can easily scale to large datasets and multiple analysts querying the system while applying the privacy system in the Audience Engagement API product at LinkedIn.  In particular, we developed a library of private algorithms and a privacy budget management system separately so that each could scale according to their own requirements; see Section~\ref{sect:Architecture} for more detail.  Further, we propose two units of budget, \emph{information} and \emph{call} budgets, that can be deducted for each analyst depending on each result she receives.  We can then use the latest, state of the art privacy composition formulas that tightly bound the overall privacy loss.  We also state our assumptions for the privacy system in Section~\ref{sect:deployment}, including analysts not colluding and data churn for refreshing privacy budgets.

\subsection{Related Work}
Differential privacy has become the standard privacy benchmark for data analytics on sensitive datasets.  Despite its popularity in the academic literature, the number of actually implemented differential privacy systems is limited, but growing.   Several of the currently implemented systems with differential privacy are in the local model, where data is individually privatized prior to being aggregated on a central server.  The main local differentially private systems include Google's RAPPOR on their Chrome browser \cite{ErlingssonPiKo14}, Apple's iOS and MacOS diagnostics \cite{ApplePrivacy17}, and Microsoft's telemetry data in Windows 10 Fall Creators Update \cite{DingKuYe17}.  

The privacy model we are interested in  for this work is the global privacy setting, where data is already stored centrally, but we want to ensure each result computed on the data is privatized.  In this less restrictive privacy setting, the main industrial differential privacy systems include Microsoft's PINQ \cite{McSherry10}, Uber's FLEX for its internal analytics \cite{JohnsonNeSo18}, LinkedIn's PriPeARL for its ad analytics \cite{KenthapadiTr18}, Google's recent differential privacy open source project \cite{Guevara2019, WilsonZhLaDeSiGi20}, and the 2020 U.S. Census \cite{DajaniLaSiKiReMaGaDaGrKaKiLeScSeViAb17}.  In this work, we present a privacy system that incorporates a privacy budget management service to ensure user-level privacy, whereas LinkedIn's PriPeARL system provides event-level privacy and was focused on providing consistent results, which we also incorporate.  The FLEX system points out that a privacy budget management service can be implemented but does not provide a strategy for how to do it.  Further, our system is part of an API that allows for adaptively chosen queries computed in real-time, which is, to our knowledge, a different model from the future U.S. Census Bureau's system.  

The main difference between the approach recently proposed in \cite{WilsonZhLaDeSiGi20} and this work is that we do not bound user contributions across and within different partitions.\footnote{Note that a caveat in the Google open-source code is that the ``implementation assumes that each user contributes only a single row to each partition.`", \url{https://github.com/google/differential-privacy}}  Such an approach would create a significant bottleneck in processing queries in a real-time system, since each online query can require a pre-processing step over the dataset to bound user contributions.  For Audience Engagement, we are dealing with terabytes of data.  The $\rT{}$ algorithm provides user-level privacy guarantees for count distinct queries without pre-processing, thus handling similar \emph{queries-per-second} (QPS) as without privacy.  See Section~\ref{sect:Tasks} for more detail.  Although Wilson et al. \cite{WilsonZhLaDeSiGi20} do discuss a privacy budget, it does not consider optimized privacy loss bounds for the data analytics tasks we consider.  Our system takes into account the various privacy algorithms to take advantage of the state of the art privacy composition bounds, such as \emph{pay-what-you-get} composition and improved composition bounds for exponential mechanisms \cite{DurfeeRo19, DongDuRo19}.  

There are other open source libraries for differentially private algorithms, such as PrivateSQL \cite{KotsogiannisTaHeFaMaHaMi19} and the recent collaboration project between Harvard's IQSS and Microsoft \cite{Kahan2019}. The former work generates a synthetic dataset, \emph{private synopses}, that is based on all queries that are posed in advance.  Such an approach is very appealing, but would not be feasible in our setting due to the size of the underlying dataset and the set of all possible queries that can be asked by an analyst also being large.  

Another related privacy system is PSI ($\Psi$) from the Harvard Privacy Tools Project \cite{PSI}.  PSI is a private data sharing interface to ``enable researchers in the social sciences and other fields to share and explore privacy-sensitive datasets with the strong privacy protections of differential privacy."  Although they support several commonly used statistics, our system covers the necessary algorithms to privatize queries in the Audience Engagement API.  Further, our system allows for handling highly distributed datasets via Pinot while enforcing a strict privacy budget that is eventually consistent across data centers.


\section{Preliminaries}

We now present some notation and fundamental definitions that will be used to describe our privacy system.  
We will denote the data histogram as $\bbh \in \N^{d}$ where $d$ is the dimension of the data universe, which might be unknown or known.  We say that $\bbh$ and $\bbh'$ are neighbors, sometimes denoted as $\bbh \sim \bbh'$, if they differ in the presence or absence of at most one member's data.   We now define differential privacy \cite{DworkMcNiSm06, DworkKeMcMiNa06}.
\begin{definition}[Differential Privacy]
An algorithm $\cM$ that takes a histogram in $\N^d$ to some arbitrary outcome set $\cY$ is $(\diffp,\delta)$-differentially private (DP) if for all neighbors $\bbh,\bbh'$ and for all outcome sets $S \subseteq \cY$, we have
$
\Pr[\cM(\bbh) \in S] \leq e^\diffp \Pr[\cM(\bbh') \in S] + \delta.
$
If $\delta=0$, then we simply write $\diffp$-DP.
\end{definition}

In our algorithms, we will add noise to the histogram counts.  The noise distributions we consider are from a Gumbel distribution where $\gum(b)$ has PDF $p_\gum(z;b)$ or a Laplace distribution where $\lap(b)$ has PDF $p_\lap(z;b)$, and
\begin{align*}
p_\gum(z;b) & = \frac{1}{b} \cdot e^{- (z/b + e^{-z/b}) } \\
p_\lap(z;b) & = \frac{1}{2b} \cdot e^{- |z | / b}.
\end{align*}

As an analyst interacts with private algorithms, the resulting privacy parameters increase with each returned result.  Hence, we need to account for the overall \emph{privacy budget} that an analyst can exhaust before the privacy loss is deemed to be too large.  We then use the composition property of DP to bound the resulting privacy parameters.  We will use \emph{bounded range} in our composition analysis, which was introduced by Durfee and Rogers \cite{DurfeeRo19}. Note that $\diffp$-BR mechanisms are $\diffp$-DP and $\diffp$-DP mechanisms are $2\diffp$-BR.  

\begin{definition}[Bounded Range]
Given a mechanism $\cM$ that takes a histogram in $\N^d$ to outcome set $\cY$, we say that $\cM$ is $\diffp$-bounded range (BR) if for any $y_1,y_2 \in \cY$ and any neighboring databases $\bbh,\bbh'$ we have
\[
\frac{\Pr[\cM(\bbh) = y_1]}{\Pr[\cM(\bbh') = y_1]} \leq e^{\diffp} \frac{\Pr[\cM(\bbh) = y_2]}{\Pr[\cM(\bbh') = y_2]}
\]
where we use the density function instead for continuous outcomes. 
\end{definition}

We now state the result from Dong et al. \cite{DongDuRo19} that tightens the composition bound from Durfee and Rogers \cite{DurfeeRo19} which itself improved on the more general optimal DP composition bounds \cite{KairouzOhVi17,MurtaghVa16}.

\begin{lemma}\label{lem:compBoundedRange}
Let $\cM_1, \cM_2, \cdots, \cM_t$ each be $\diffp$-BR where the choice of mechanism $\cM_i$ at round $i$ may depend on the previous outcomes of $\cM_1, \cdots, \cM_{i-1}$, then the resulting composed algorithm is $(\diffp'(\delta), \delta)$-DP for any $\delta \geq 0$ where $\diffp'(\delta)$ is the minimum of $t \diffp$ and
\begin{equation}\label{eq:compBetter}
t\left( \frac{\diffp}{1 - e^{-\diffp}} - 1 - \log \left( \frac{\diffp}{1 - e^{-\diffp}} \right) \right)+  \diffp \sqrt{ \frac{t}{2} \log(1/\delta)}. 
\end{equation}
\end{lemma}

We also can use the more complicated composition bound for BR mechanisms \cite{DongDuRo19}.  However we cannot use the optimal composition bound from \cite{DongDuRo19} because it only applies to the non-adaptive setting.  Here we are interested in the API setting which allows the user to ask adaptive queries, meaning the queries can depend on previous results.


\section{Private Data Analytics}\label{sect:Tasks}

To incorporate differential privacy, we needed to consider the various tasks we want the application to handle.  We will be focusing on data analytics based on histograms or counts over different domain elements.  We will discuss each query type our privacy system handles, but first we need to set up some notation.   

In order to provide a \emph{user-level} privacy guarantee where all data records of a user are protected, as opposed to \emph{event-level} where only an individual data record is protected, we consider two types of queries.  The first consists of \emph{distinct count} queries where a member can contribute a count of at most 1 to any number of elements, i.e. $||\bbh - \bbh'||_\infty \leq 1$ for any neighbors $\bbh,\bbh'$ ($\ell_\infty$-sensitivity).  An example of such a query would be ``what are the top-$k$ articles that are shared among distinct members with a certain skill set?"  The second type is \emph{non-distinct count} queries where a member can increase the count of any element by at most $\tau\geq 1$, i.e. $||\bbh - \bbh' ||_\infty \leq \tau$ for any neighbors $\bbh, \bbh'$.
Note that $\tau = 1$ gives us the distinct count setting and $\tau$ can be a parameter for each non-distinct count query.

In either case, distinct count or non-distinct count queries, a member can either affect the count of an arbitrary number of elements $||\bbh - \bbh'||_0 \leq d$ for any neighbors $\bbh,\bbh' \in \N^d$ or a bounded number $||\bbh - \bbh'||_0 \leq \Delta$ for $\Delta < d$ ($\ell_0$-sensitivity).  We separate these two cases as the \emph{unrestricted} sensitivity and $\Delta$-restricted sensitivity settings, respectively.  In the case of unrestricted sensitivity, we will return only a fixed number of counts, say the top-$k$, in order to bound the privacy loss.

Scaling our privacy system across several analysts with queries that require data from multiple servers requires algorithms that can run efficiently with runtime that does not scale with the entire data domain size $d$.  For example, for the top-10 articles engaged with by staff software engineers, we do not want to query over all articles, since there could potentially be billions of articles and would be computationally expensive and slow.  For this reason, we distinguish the case where the data domain is reasonably sized and known, i.e. \emph{known domain}, from when the data domain is very large or unknown, i.e. \emph{unknown domain}. 

We then summarize in Table~\ref{table:tasks} the set of queries that we want our privacy system to handle into \emph{unrestricted sensitivity} or $\Delta$-\emph{restricted sensitivity} as well as \emph{known domain} or \emph{unknown domain} with the corresponding algorithms we will use for each setting.  Recall that we can interpolate between distinct count queries and non-distinct count queries with the $\tau\geq 1$ parameter, so we include $\tau$ as a parameter to each of our algorithms.  Furthermore, each algorithm takes a privacy parameter $\diffpper$.
\begin{table}[htbp]
\centering\setcellgapes{4pt}\makegapedcells
\begin{tabular}{ |c|c|c| } 
 \hline
 & $\Delta$-restricted sensitivity & unrestricted sensitivity \\ 
 \hline
 \shortstack{Known \\ Domain} & $\knownLap{\Delta,\tau}$ \cite{DworkMcNiSm06} & $\knownEM{k,\tau}$ \cite{McSherryTa07} \\ 
 \hline
\shortstack{Unknown \\ Domain} & $\rTE{\Delta,\bar{d},\tau}$ & $\rT{k,\bar{d},\tau}$ \\ 
 \hline
\end{tabular}
\caption{DP algorithms for various data analytics tasks\label{table:tasks}}
\end{table}

Restricting the $\ell_\infty$-sensitivity is not part of the algorithm, rather it is done by using distinct count queries ($\tau =1$), knowing a bound a priori, or done via a preprocessing step on the data.  For the unknown domain setting, we require a parameter $\bar{d}$ which tells us how many elements from the original dataset that we can access in our algorithms.  We think of $\bar{d}$ as the maximum number of elements our OLAP system can return without causing significant latency.  For the unrestricted sensitivity setting, we require our algorithms to return at most $k$ elements, such as the top-$k$.  This is due to the fact that a user can change the counts of an arbitrary number of elements.  In such cases, to have any hope to bound the privacy loss, we bound the number of elements that can be returned.  In Table~\ref{table:tasks}, we refer to the following mechanisms: the standard Laplace mechanism from \cite{DworkMcNiSm06} is denoted as $\knownLap{\Delta,\tau}$, which adds Laplace noise with scale proportional to $\tau/\diffpper$ to each count; the $k$-peeling exponential mechanism \cite{McSherryTa07}, which adds Gumbel noise with scale proportional to $\tau/\diffpper$ to each count and returns the elements with the top-$k$ noisy counts, which we denote as $\knownEM{k,\tau}$;  the generalized restricted sensitivity algorithm from \cite{DurfeeRo19} denoted as $\rTE{\Delta,\bar{d},\tau}$, which is presented in Algorithm~\ref{algo:genLimitDomLap}; the generalized unrestricted sensitivity algorithm from \cite{DurfeeRo19} denoted as $\rT{k,\bar{d},\tau}$, which is presented in Algorithm~\ref{algo:genLimitDom}.   

The primary difference between querying the OLAP datastore for results with privacy as opposed to without privacy is that when querying for the top-$k$ in the unknown domain setting, we instead fetch the top-$\bar{d}$ and then use $\rT{k,\bar{d},\tau}$ or $\rTE{\Delta,\bar{d},\tau}$.  Ideally, we would want to set $\bar{d} = d$ to get the full dataset, but that is not practical when the number of elements is large and the existing architecture potentially trims the results for efficiency.  The choice of algorithm for each query can be a simple look up of the \emph{group by} clause where if no additional information is given, then we default to the unknown domain and unrestricted sensitivity.


\section{Privacy System Architecture\label{sect:Architecture}}
Existing OLAP datastores are designed to provide real-time data analytics over distributed datasets, with differential privacy not necessarily being incorporated from the beginning.  Pinot is the analytics platform of choice at LinkedIn for site-facing use cases.   In this section, we detail how we incorporated differential privacy with Pinot and the application. Figure~\ref{fig:PrivSystemOverview} presents the overall system.

\begin{figure}[h]
\centering
\includegraphics[width=\columnwidth]{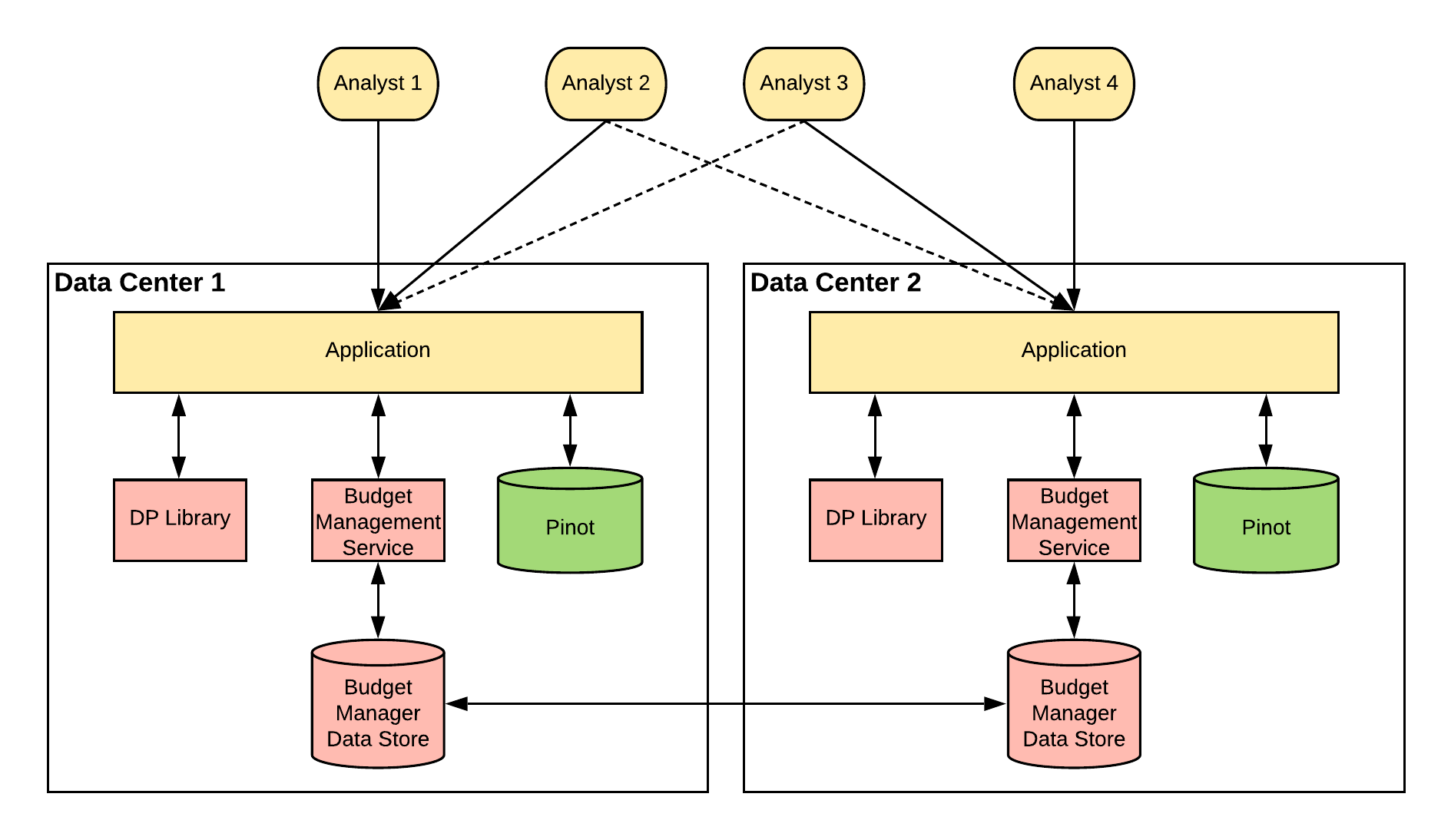}
\caption{The overall privacy system with additional components for DP being the \emph{DP Library} as well as the \emph{Budget Management Service} and \emph{Data Store}. The arrows between \emph{Analysts} and \emph{Data Centers} show that an analyst may be initially assigned one data center (bold) but can migrate to a different one (dashed).\label{fig:PrivSystemOverview}}
\end{figure}

The application entity, based on the request received from the analyst, generates queries to the underlying database. The queries typically ask for a histogram grouped by some column. In order to apply the right algorithm for the query, the application needs to know the sensitivity and domain setting of the column as shown in Table~\ref{table:tasks}. Also, the query is to be modified to fetch a potentially larger number of rows from the database.

We designed generic interfaces that are implemented by a suite of algorithms. The interfaces allow the application to:
\begin{itemize}
\setlength\itemsep{0em}
\item Retrieve modified query parameters (e.g. change $k$ to $\bar{d}$).
\item Estimate privacy cost of the query (e.g. return $\Delta$ or $k$).
\item Add noise to results based on configured parameters.
\item Compute the actual cost of the query, e.g. the number of items returned in the unrestricted sensitivity setting.
\end{itemize}

The application can independently invoke budget management functions, such as the following:

\begin{itemize}
\setlength\itemsep{0em}
\item Getting the available budget for an analyst to verify whether a query can even start to execute.
\item Depleting the available budget with the executed query.
\end{itemize}
The cost of a query could be multi-dimensional, including the cost of making the call and of information retrieved (see Section~\ref{sect:Budget}). 

Given the query from the analyst and the selected DP algorithm, the application will then interact with the DP library.  It will first determine the expected cost of the resulting query to show the application, which is a function of the query that is asked and the selected algorithm.   The application then calls the DP library to translate the query to a DP version that will be used to query Pinot.  For example, if the query is for top-$k$ and the algorithm is in the unknown domain setting, then the translation could simply modify $k$ to $2k$, in which case $\bar{d} = 2k$ in $\rTE{\Delta,\bar{d},\tau}$ and $\rT{k,\bar{d},\tau}$.  On the other hand, if the query is over the known domain setting, then we will want to translate $k$ to $d$ in order to get counts over the full domain, including elements with zero counts.

Now that the application has the modified query from the DP library, we need to check whether there is enough budget remaining from the budget management service for the query to be evaluated, which are updated parameters $(\kmax,\ellmax)$ that decrease from some fixed values.  We typically call the $\kmax$ parameter the \emph{information budget} and can be thought of as the amount of the $\diffp$ parameter in DP we are consuming.  Additionally, we refer to $\ellmax$ as the \emph{call budget} and is associated with the $\delta$ parameter in DP.  We assume that each analyst will have its own budget and each analyst starts with the same budget.  If the budget is exhausted for an analyst then the budget management service does not allow the query to be executed and tells the application that the analyst has exhausted their entire budget.  If the budget is not depleted, yet what remains is less than the expected cost of the query, then we still do not evaluate the query.

Once the budget management service allows for the query to be evaluated, the application queries Pinot as it would have without the privacy system only now with the translated query.  The Pinot result is then returned to the application and then the application makes another call to the DP library with the Pinot result.  The DP library will then run the corresponding DP algorithm on the Pinot result and return the DP result.  Based on the DP result, the budget management service updates the parameters $\kmax,\ellmax$, as described in Algorithm~\ref{algo:Budget}, and returns the result to the application.  

We built the algorithms module and the budget management module to be independent of each other for the following reasons:
\begin{itemize}
\setlength\itemsep{0em}
\item While DP algorithms are running on the application layer, budget management operations require a remote call to a distributed system because the budget management service needs to provide a consistent view to all application instances. Therefore, keeping the budget management independent of algorithms will allow us to scale them independently. The algorithms will need to scale to minimize memory and CPU usage, whereas the budget management service will need to scale in terms of handling higher query-per-second (QPS), while minimizing latency.
\item We require that newer (as yet unknown) algorithms still be able to use and manage budgets. 
\item Multiple implementations of the budget manager are possible depending on system requirements. We need to be able to iterate on these independently and quickly.
\item The algorithms need not (and do not) know about the analyst that is querying, and the budget manager does not behave differently depending on the type of query or algorithms used. As long as they both have a common notion of the units $(\kmax,\ellmax)$ and dimensions of cost, it makes sense to keep these independent of each other.
\end{itemize}

\subsection{Pinot: a distributed OLAP datastore}
Pinot \cite{PinotSigmod} is a distributed, real-time, columnar OLAP data store, currently incubating in Apache. At LinkedIn, we have two main categories of analytics applications: internal applications (such as dashboards, anomaly detection platform, A/B testing, etc.)  and site-facing applications (such as \emph{Who viewed my profile},  Talent Insights, etc.). Internal dashboards need to process a large volume of data (trillions of records), but can tolerate latencies in hundreds of milliseconds. They also have a relatively low query volume. The site-facing applications, on the other hand, serve hundreds of millions of LinkedIn members, and therefore have a very high query volume with a latency budget of a few to perhaps tens of milliseconds. 

Pinot has a flexible architecture and supports a wide variety of applications in the spectrum. Pinot production clusters at LinkedIn are serving tens of thousands queries per second, supporting more than 50 analytical use cases, and ingesting over millions of records per second.
Other companies such as Uber, Microsoft, and Weibo are also operating production Pinot clusters.

\begin{figure}[h]
\centering
\includegraphics[width=\columnwidth]{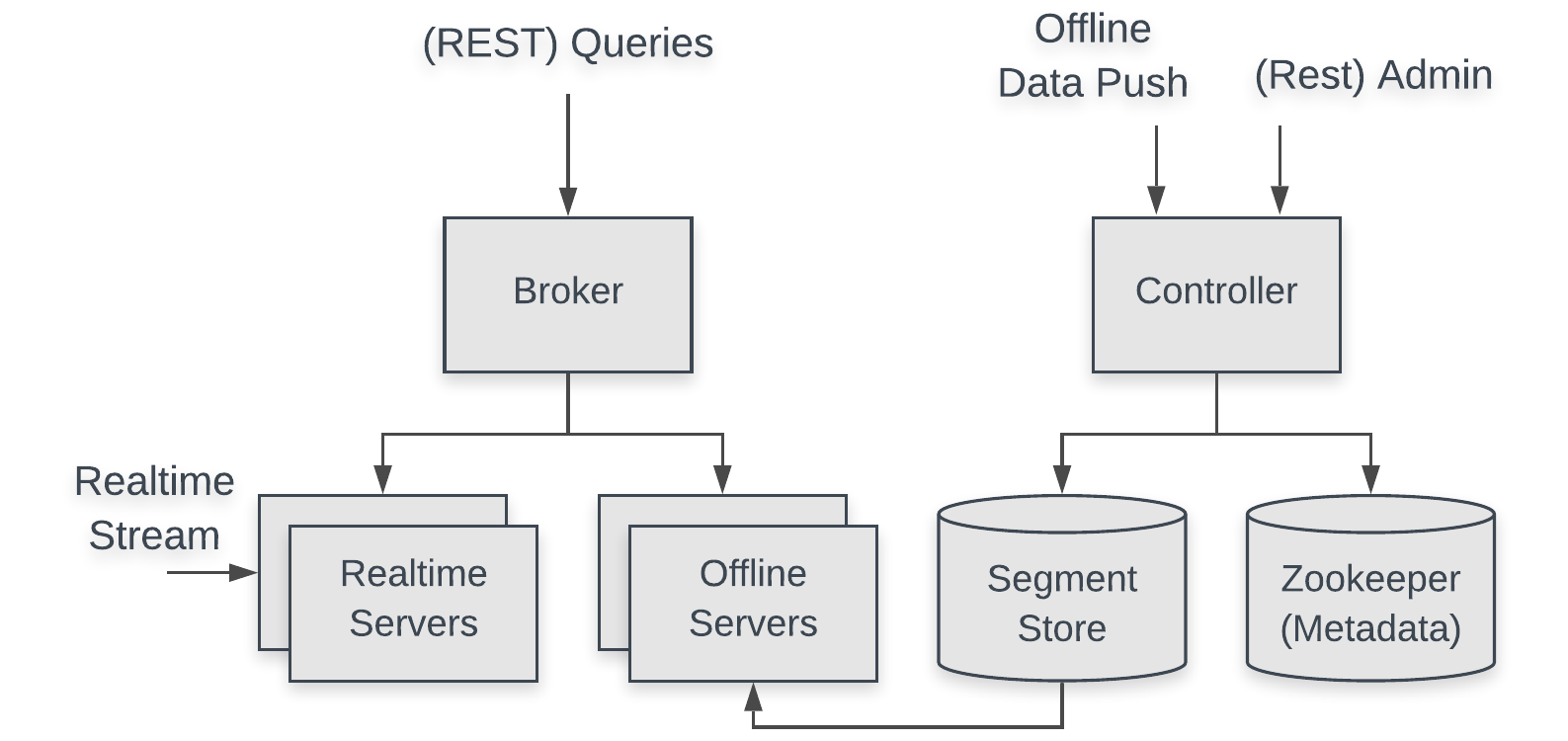}
\caption{Overall Pinot Architecture. \label{fig:PinotSystemOverview}}
\end{figure}

As shown in Figure~\ref{fig:PinotSystemOverview}, Pinot has three different components: controller, broker, and server. Controllers handle cluster wide coordination, run periodic tasks for cluster state validation and retention management, and provide a REST API for managing cluster metadata. Brokers receive queries and federate them to servers so as to cover all the segments (shards) of a table.  Servers execute the query on the segments.  Offline servers host segments that are batch ingested while real-time servers host the segments that are ingested from streaming sources, such as Kafka \cite{Kafka}.  

For the privacy system at LinkedIn, we naturally decided to use Pinot as an OLAP data store because Pinot already supported a lot of customer-facing analytics applications like Audience Engagement. However, it is noteworthy that our architecture keeps the the budget management service and Pinot as separate components so that we can easily provide DP features to other analytical query engines such as Presto and Spark SQL.

\subsection{Key-value Based Budget Management System in Espresso}

We now describe our key-value based budget management system.  We create one key per analyst of a table (or per use case which may have multiple tables), and the data against the key is atomically changed when we need to update the budget. The store needs to provide ways to do the read-modify-write operations, and the latency should be relatively low.

The value record will contain the following items:
\begin{itemize}
\item Maximum budget allowed for the user.
\item The time period over which this budget is allowed (typically the time period during which the data is refreshed completely).
\item The total budget used so far (or, that remain).
\item The timestamp when the used budget was reset to 0 (e.g. for a monthly refresh, this will be the 1st of the month).
\end{itemize}

There are three methods that the budget manager needs to support:
\begin{itemize}
\item To check whether an analyst's ID has enough budget to run a query that will consume at most a given cost, we use $\texttt{checkBudget}(\text{ID}, \text{cost} )$, which returns either true or false.
\item To deduct an analyst's budget, with a given ID, after getting a DP result, with a given cost, we use $\texttt{updateBudget}(\text{ID}, \text{cost})$.
\item To get the current budget of an analyst, with a given ID, we use $\texttt{getBudget}(\text{ID})$, which returns either the analyst's current budget that has been used or the maximum budget allowed if there has been a budget refresh.
\end{itemize}

We use an \emph{Espresso} \cite{Espresso13} key-value store to manage the budget. The read requests should be fast, given it is only a primary key lookup. So, a call to get the current usage (currently coming in at an unknown rate) can be fast.  Espresso was chosen due to several reasons: eventual consistency in cross-datacenter replication to ensure an analyst does not exceed a given budget, capability to scale to millions of users while still keeping a fairly constant response time, control over the refresh time period, and flexibility to change per-analyst maximum upon demand.


\section{Differentially Private Algorithms}

We detail the algorithms for the various tasks in Table~\ref{table:tasks}.  These algorithms consist of previous work from \cite{DworkMcNiSm06}, \cite{McSherryTa07}, and \cite{DurfeeRo19}, or slightly modified forms.  Each algorithm takes a $\diffpper$ privacy parameter, which determines the amount of noise to add, while each algorithm in the unknown domain setting has an additional $\delta >0$ privacy parameter.  We point out that $\rT{k,\bar{d},\tau}$ is the default algorithm to use when no other information is known.  However, the benefit of knowing the domain is that when $k$ results are requested, $k$ results will be returned each time, whereas the unknown domain setting may return fewer than $k$.  The benefit of the $\Delta$-restricted sensitivity setting is that the budget depletes by only $\Delta$, rather than by the number of elements returned, from  as in the unrestricted setting.

\subsection{Known Domain Algorithms}
We will now state the well known Laplace \cite{DworkMcNiSm06} and Exponential \cite{McSherryTa07} mechanisms. 
We present the Laplace mechanism  \cite{DworkMcNiSm06} in the context of histogram data with the assumption that the $\ell_\infty$-sensitivity between any neighbors is bounded by $\tau$.  Note that we will use a slightly different scale of noise in procedure $\knownLap{\Delta,\tau}$ in Algorithm~\ref{algo:LapMech} than is traditionally used.  This is because we want to compose bounded range algorithms in our privacy budget manager, where each algorithm has the same parameter $\diffpper$.  We go into more detail of the privacy budget service in Section~\ref{sect:Budget}.  

\begin{algorithm}[h!]
	\caption{$\knownLap{\Delta,\tau}$; Laplace mechanism over known domain with $\ell_\infty$-sensitivity $\tau$, and $\Delta$-restricted sensitivity}
	\begin{algorithmic} 
		\State \textbf{Input:} Histogram $\bbh$, $\Delta$ sensitivity, along with parameter $\diffpper$.
		\State \textbf{Output:} Noisy histogram.
		\For{$i \in [d]$}
			\State $v_i = h_i + \lap(2\tau/\diffpper)$
		\EndFor
		\State Return $\{v_1, \cdots, v_d \}$	
	\end{algorithmic}\label{algo:LapMech}
\end{algorithm}

We now discuss the Exponential Mechanism \cite{McSherryTa07} in full generality and use the \emph{range} of a quality score rather than the global sensitivity of the score, as was presented in \cite{DongDuRo19}.

\begin{definition}[Exponential Mechanism] The \emph{Exponential Mechanism} $M_q: \cX \to \cY$ with quality score $q: \cX \times \cY \to \R$ is written as $M_q(x)$, which samples $y$ with probability proportional to $\exp\left( \diffpper q(x,y) / S_q \right)$ where,
\begin{align*}
& S_q\defeq 
\\
& \sup_{x \sim x'} \left\{  \max_{y\in \cY} \{ q(x,y) - q(x',y) \} - \min_{y'\in \cY} \{q(x,y') - q(x',y') \}\right\}.
\end{align*}
\end{definition}
Note that the Exponential Mechanism is equivalent to adding Gumbel noise $\gum(S_q/\diffpper)$ to $q(x,y)$ for each $y \in \cY$ and reporting the largest noisy counts \cite{DurfeeRo19}.
We then have the following result from \cite{McSherryTa07, DongDuRo19}
\begin{lemma}
The Exponential Mechanism is $\diffpper$-BR and, hence $\diffpper$-DP.  
\end{lemma}
In our case, the quality score will simply be the heights of the histogram.  Note that we have only discussed the Exponential Mechanism to return a single element.  In the case where we want to return $k$-elements, we can iteratively apply the Exponential Mechanism by removing the element that is returned in each round and then run the Exponential Mechanism again without the previously returned elements, also known as \emph{peeling}.  However, we can implement this more efficiently by adding Gumbel noise to all the counts and then releasing the top-$k$ elements in a single shot \cite{DurfeeRo19}.  However, we need to also include counts, so we add independent Laplace noise to the counts of the elements in the noisy top-$k$.  We then formally present the $\knownEM{k,\tau}$ procedure in Algorithm~\ref{algo:EM}.

\begin{algorithm}[h!]
	\caption{$\knownEM{k, \tau}$; Exponential Mechanism over known domain with $\ell_\infty$-sensitivity $\tau$ and unrestricted sensitivity}
	\begin{algorithmic} 
		\State \textbf{Input:} Histogram $\bbh$, number of outcomes $k$, and parameter $\diffpper$.
		\State \textbf{Output:} Ordered set of indices and counts.
		\For{$i \in [d]$}
			\State $v_i = h_i + \gum(\tau/\diffpper)$
		\EndFor
		\State Sort $\{v_i \}$ where $v_{i_{1}} \geq \cdots \geq v_{i_{d}}$
		\State $\{Z_i \}_{i=1}^k \stackrel{i.i.d.}{\sim} \lap(2\tau/\diffpper)$
		\State Return $\left\{ (i_{1}, h_{i_{1}} + Z_1), \cdots,(i_{k}, h_{i_{k}} + Z_k )\right\}$	
	\end{algorithmic}\label{algo:EM}
\end{algorithm}

We then have the following result which follows from Dwork et al. \cite{DworkMcNiSm06}, as well as from McSherry and Talwar \cite{McSherryTa07, DongDuRo19}.
\begin{lemma}
Assume that $||\bbh - \bbh'||_\infty \leq \tau$ and $|| \bbh - \bbh'||_0 \leq \Delta$ for any neighbors $\bbh,\bbh'$. The procedure $\knownLap{\Delta,\tau}$ is $\Delta\diffpper/2$-DP and $\Delta\diffpper$-BR.  Further, if $\Delta$ is large or unknown then $\knownEM{k,\tau}$ is $3k \diffpper/2$-DP and $2k\diffpper$-BR.
\end{lemma}

\subsection{Unknown Domain with $\Delta$-Restricted Sensitivity\label{sect:restricted_sensitivity}}
\begin{algorithm}[H]
	\caption{$\rTE{\Delta,\bar{d},\tau}$; Laplace mechanism over unknown domain with access to $\bar{d} +1 > \Delta$ elements, $\ell_\infty$-sensitivity $\tau$, and $\Delta$-restricted sensitivity.}
	\begin{algorithmic} 
		\State \textbf{Input:} Histogram $\bbh$, $\Delta$ sensitivity, cut off at $\bar{d} + 1$, and $\diffpper,\delta$.
		\State \textbf{Output:} Ordered set of indices and counts.
		\State Solve for $\hat{\delta}$: 
		$
		\delta = \hat\delta/4 \cdot (e^{\diffpper/2} + 1)(3 + \ln(\Delta/\hat\delta))
		$
		\State Sort $h_{(1)} \geq h_{(2)} \geq \cdots \geq h_{(\bar{d} +1 )}$.
		\State  $v_\bot = h_{(\bar{d} + 1)}+ \tau \cdot (1 + 2\Delta\ln(\Delta/\hat\delta)/\diffpper) + \lap(2\tau\Delta/\diffpper)$
		\For {$i \leq \bar{d}$}
        \State Set $v_i = h_{(i)} + \lap(2\tau\Delta/\diffpper)$
		\EndFor
        \State Sort $\{v_i\} \cup v_\bot$
        \State Let $v_{i_{1}},....,v_{i_{j}}$ be the sorted list until $v_\bot$
		\State Return $\{(i_{1},v_{i_{1}}),....,(i_{j},v_{i_{j}) }, (\bot,v_\bot) \}$.		
	\end{algorithmic}\label{algo:genLimitDomLap}
\end{algorithm}
For our unknown domain algorithms, we introduce a $\bot$ character to denote a null element that is not part of the domain and whose count is a noisy threshold where no element with smaller noisy count is returned.  We present the $\rTE{\Delta,\bar{d},\tau}$ procedure in Algorithm~\ref{algo:genLimitDomLap} in a more general form than in \cite{DurfeeRo19}, which only considered the distinct count case, i.e. $\tau = 1$.  Further, the proof of privacy remains true if we release the counts as well as the indices. For completeness, the proof of the following result is presented in the appendix.  
\begin{lemma}[Durfee and Rogers \cite{DurfeeRo19}]\label{lem:DDR19}
Assume that $||\bbh - \bbh'||_\infty \leq \tau$ and $|| \bbh - \bbh'||_0 \leq \Delta$ for any neighbors $\bbh,\bbh'$, then the procedure $\rTE{\Delta,\bar{d},\tau}$ is $(\diffpper/2,\delta)$-DP.
\end{lemma}

\subsection{Unknown Domain with Unrestricted Sensitivity\label{sect:unrestricted}}
\begin{algorithm}[h!]
	\caption{$\rT{k,\bar{d},\tau}$; Unknown domain mechanism with access to $\bar{d} + 1 > k$ elements, $\ell_\infty$-sensitivity $\tau$}
	\begin{algorithmic}
		\State \textbf{Input:} Histogram $\bbh$; outcomes $k$, cut off at $\bar{d}+1$, and $\diffpper,\delta$.
		\State \textbf{Output:} Ordered set of indices and counts.
		\State Sort $h_{(1)} \geq h_{(2)} \geq \cdots \geq h_{(\bar{d} + 1)}$.
		\For {$i \in \{k, \cdots, \bar{d} \}$}
			\State Set $v_i = h_{(i+1)} + \tau + \tau \log(i/\delta)/\diffpper + \gum(\tau/\diffpper)$
		\EndFor
		\State Set $\bk = \argmin\{v_i \}$.
		\State Set $h_{\bot} = h_{(\bk+1)} + \tau \cdot (1 + \ln(\min\{\bk, \bar{d}- \bk\}/\delta)/\diffpper)$.
		\State Set $v_{\bot} = h_{\bot} +  \gum(\tau/\diffpper)$.
		\For {$j \leq \bk$}
				\If{ $h_{(j)} > h_{(\bar{k} + 1)}$ }
        					\State Set $v_{(j)} = h_{(j)} + \gum(\tau/\diffpper)$.
				\EndIf
		\EndFor
        		\State Sort $\{v_{(j)}\} \cup v_{\bot}$.
        		\State Let $v_{i_{1} },....,v_{i_{j}},v_{\bot}$ be the sorted list up until $v_{\bot}$.
		\State $\{Z_i \}_{i=1}^j \stackrel{i.i.d.}{\sim} \lap(2\tau/\diffpper)$
		\If{$j < k$}
			\State Return $\{(i_{1}, h_{i_{1}} + Z_1 ),...,(i_{j}, h_{i_{j}} + Z_j ),\bot\}$ 
		\Else
			\State Return $\{(i_{1}, h_{i_{1}} + Z_1 ),...,(i_{k}, h_{i_{k}} + Z_k )\}$.	
		\EndIf
	\end{algorithmic}\label{algo:genLimitDom}
\end{algorithm}
We present the $\rT{k,\bar{d},\tau}$ procedure in Algorithm~\ref{algo:genLimitDom} in a more general form than in \cite{DurfeeRo19}, which only considered the distinct count case, i.e. $\tau = 1$.  The proof of the following theorem follows the same analysis as in \cite{DurfeeRo19}.  Note that we use the optimal threshold index procedure from Algorithm 6 in Durfee and Rogers \cite{DurfeeRo19} by default and return counts by adding Laplace noise to the discovered elements in the top-$k$.
\begin{theorem}[Durfee and Rogers \cite{DurfeeRo19}]
Assume $||\bbh - \bbh'||_\infty \leq \tau$ for any neighbors $\bbh,\bbh'$. Then $\rT{k,\bar{d},\tau}$ is $((2k+1)\diffpper,\delta)$-DP.
\end{theorem}


\section{Privacy Budget Management Service \label{sect:Budget}}
We ultimately want to ensure that no analyst can identify any individual's data with high confidence.  We then impose a strict overall $(\diffpmax,\deltamax)$-DP guarantee.  In order to compute the parameters $(\diffpper,\delta)$ that we use in each call to our algorithms\footnote{Note that for the Laplace and Exponential mechanisms in the known domain, $\delta =0$.} over an entire sequence of interactions with the API, we also want to know how many queries the API will allow, denoted as $\ellmax$ that we term the \emph{call budget}, which will effectively impact $\deltamax$.  Further, we want to track the number of elements we want to return, denoted as $\kmax$ that we term the \emph{information budget}, which will effectively impact $\diffpmax$.  Note that $\kmax$ does not necessarily equal the number of elements returned, because we might be in the restricted sensitivity setting, and $\ellmax$ does not precisely equal the number of calls to the API, since we might be in the known domain setting for some queries.  Once we have $(\diffpper,\delta)$, we will only use these parameters in each algorithm, hence not allowing for adaptively changing privacy parameters.  

\subsection{Budget Management Implementation}

As mentioned in Section~\ref{sect:Architecture}, the budget manager needs to be a distributed system so that it can be accessed/updated from different application execution platforms.  Each analyst may access data from multiple data centers and each access must deduct from the same budget. Hence, the budget manager maintains eventual consistency across data centers.  

The budget can be thought of as an associative array with keys from $[\ell, k]$ and values as the corresponding units used.
Given a particular outcome $o$ from the API, the budget service will update $\kmax \gets \kmax - \Delta$ for $\Delta$-restricted sensitivity queries or $\kmax \gets \kmax - 2|o|$ for unrestricted sensitivity queries where $|o|$ denotes the number of elements returned in outcome $o$.  Furthermore, the privacy budget management system will update $\ellmax \gets \ellmax - 1$ for each query the analyst makes that is in the unknown domain setting.  Once $\kmax$ or $\ellmax$ are depleted, we prevent the analyst from making any other queries.  We address the challenge of computing the individual privacy parameters $(\diffpper,\delta)$ given $(\diffpmax,\delta^\star, \kmax,\ellmax)$ in Theorem~\ref{thm:final_parameters}.

We adopt a privacy budget management service that assumes any user does not collude with other analysts.  Hence each analyst is given her own privacy budget to interact with the Audience Engagement API and her queries do not impact the budget of another analyst.  One can imagine variants of this assumption, such as all analysts that belong to the same company must share a budget.  Further, the API adheres to the privacy budget up to some time frame.  Thus, if an analyst has asked more than $\ellmax$ unknown domain queries, then she will not be allowed any further queries.  After this prescribed time frame, the parameters effectively get refreshed and the analyst can continue asking queries.  Refresh is acceptable at regular intervals if the underlying data is flushed and replaced at similar intervals, whether through complete snapshot replacement, or rolling windows, such that the user's data does not remain constant.

The application links with a budget manager client library so as to hide the implementation details of the budget management service because application writers do not need to know the details about the budget database, or the budget refresh mechanisms. The parameters ($\kmax,\ellmax$) may be configured by the application.

\subsection{Differential Privacy Composition\label{sect:privacybudget}}

We present pseudocode for the privacy budget management service in Algorithm~\ref{algo:Budget}.  We then present a way to compute the privacy guarantee of our overall system, which largely follows the analysis from Durfee and Rogers \cite{DurfeeRo19}.  Essentially, the analysis follows from the fact that each algorithm can be represented as an iterative sequence of $\diffpper$-BR algorithms.  Note that the algorithms in the unknown domain setting have a probability $\delta$ of larger privacy loss, which we account for in the overall $\delta^\star$ in the privacy guarantee.
\begin{algorithm}[h!]
	\caption{$\BMS{\kmax,\ellmax}$; Budget Management}
	\begin{algorithmic}
		\State \textbf{Input:} An adaptive stream of histograms $\bbh_{1},\bbh_{2},....$, 	fixed integers $\kmax$ and $\ellmax$, along with per iterate privacy parameters $\diffpper,\delta$.
		\State \textbf{Output:} Sequence of outputs $(o_{1}, o_{2}, \cdots)$. 
		\While {$\kmax > 0$ and $\ellmax > 0$}
		\State Select $\bbh_{i} \in \N^{d_i}$ with $\ell_\infty$-bound $\tau_i$.
		\State Select $k_i$ and number of elements allowed to access $\bar{d}_i$
		\If{ histogram has $\Delta$-restricted sensitivity}
			\If{$\Delta> \kmax$}
				\State Break
			\EndIf
			\If{$\bar{d}_i > d_i$}
				\State $o_i =\knownLap{\Delta,\tau_i}(\bbh_i)$.
				\State Update $\kmax \gets \kmax - \Delta$.
			\Else
				\State $o _i = \rTE{\Delta,\bar{d}_i,\tau_i}(\bbh_i)$
				\State Update $\ellmax \gets \ellmax -1$. 
				\State Update $\kmax \gets \kmax - 1$.
			\EndIf
		\EndIf
		\If{ histogram has unrestricted sensitivity}
			\If{$2k_i > \kmax$}
				\State Break
			\EndIf
			\If{$\bar{d}_i > d_i$}
				\State $o_i =\knownEM{k_i,\tau_i}(\bbh_i)$.
				\State Update $\kmax \gets \kmax - 2k_i$.
			\Else
				\State $o _i = \rT{k_i,\bar{d}_i,\tau_i}(\bbh_i)$
				\State Update $\ellmax \gets \ellmax -1$.
				\State Update $\kmax \gets \kmax - (2|o_i| +1 - \1{o_i[-1] = \bot})$.
			\EndIf
		\EndIf
		\EndWhile
		\State Return $o = (o_{1},o_{2},\cdots)$	
	\end{algorithmic}\label{algo:Budget}
\end{algorithm}

In order to allow for the budget management service to return counts in the unrestricted sensitivity setting, we need to account for that in our overall budget.  Further, in the unknown domain/unrestricted sensitivity setting, if the last element of $o_i$, denoted as $o_i[-1]$, is $\bot$ at round $i$ then adding Laplace noise with parameter $2\tau_i/\diffpper$ to the counts of each of the discovered $|o_i| - 1$ elements. will ensure $\diffpper$-BR for each count.  We can then apply our privacy loss bounds to get an overall DP guarantee by updating $\kmax \gets \kmax - 2|o_i|$ and when the last element in $o_i$ is not $\bot$, then we instead update $\kmax \gets \kmax- (2|o_i| + 1)$.  Note that if we did not require counts in the results and need only return an ordered list of elements in the top-$k$, then we need only update $\kmax \gets \kmax - (|o_i|+1)$.

\begin{theorem}\label{thm:final_parameters}
For $\delta' \geq 0$ and $\diffpper,\delta >0$, the $\BMS{\kmax,\ellmax}$ is $(\diffpmax,\delta^\star)$-DP where $\delta^\star = 2\ellmax\delta + \delta'$ and $\diffpmax$ is defined as the minimum between $\kmax\diffpper$ and the following,
\begin{align}
& \kmax\left( \frac{\diffpper}{1 - e^{\diffpper}} - 1 - \log\left( \frac{\diffpper}{1 - e^{\diffpper}} \right) \right) \nonumber \\
& \qquad\qquad\qquad +\diffpper  \sqrt{ \frac{\kmax}{2} \log(1/\delta')}.
\label{eq:epsmax}
\end{align}
\end{theorem}
\begin{proof}
For the $\Delta$-restricted sensitivity setting, we are deducting the information budget by $\Delta$ in the known domain setting or we scale the privacy parameter by $\Delta$ and deduct one from the information budget in the unknown domain setting.  For a given histogram in the known domain setting, adding $\lap(2\tau/\diffpper)$ to each count will ensure $\Delta\diffpper$-BR.  We can also analyze this mechanism as if we iteratively add $\lap(2\tau/\diffpper)$ to each count and then apply any DP or BR composition bound to obtain a DP guarantee.  We need only apply composition for the number of elements that can actually change between neighboring datasets, i.e. $\Delta$, and not the full dimension of the histogram.  For all settings, the application of each Laplace mechanism is $\diffpper/2$-DP, hence $\diffpper$-BR, while each application of Gumbel noise (Exponential Mechanism) is $\diffpper$-BR.  Thus, we apply the composition bounds for BR mechanisms from Dong et al. \cite{DongDuRo19}.  The resulting bound applies BR composition over $\kmax$ many $\diffpper$-BR mechanisms and deducting 1 from the call cost budget $\ellmax$ if an unknown domain algorithm is used for a query.
\end{proof}
Given the total budget for the number of outcomes and queries $(\kmax,\ellmax)$ along with privacy budget $(\diffpmax,\deltamax)$ we can solve for the parameter $\diffpper$ that satisfies the budget, which is then used in each algorithm.  One approach we can use is the following (somewhat arbitrary) choice for $\delta = \tfrac{\delta^\star}{4 \ellmax}$ and $\delta' = \delta^\star/2$.


\section{Results}
We now present some preliminary results of our privacy system for the Audience Engagement API.  In Figure~\ref{fig:AlgoUnkGum} we present curves for the number of discovered elements in a top-$50$ query with varying $\diffpper$ and $\bar{d}$, i.e.  the number of elements to collect, in procedure $\rT{50,\bar{d},1}$ from Algorithm~\ref{algo:genLimitDom} with a fixed $\delta = 10^{-10}$.  The query is to find the top articles that distinct members from the San Francisco area are engaging with.  We provide intervals that contain the 25th and 75th percentiles over 1000 independent trials.  Note that the randomness in each trial is solely from the noise generation and we are using the same dataset each time.   We see that with the same level of privacy, increasing the number of elements to fetch allows us to \emph{discover} more elements.  Hence, we see a natural tradeoff not just between privacy $(\diffpper)$ and utility (number of elements returned), but also between run time (fetching more results) and utility.  For example, we can return twice as many elements if we fetch four times more elements with Pinot and setting $\diffpper = 0.08$.

\begin{figure}[h]
\centering
\includegraphics[width=.5\textwidth]{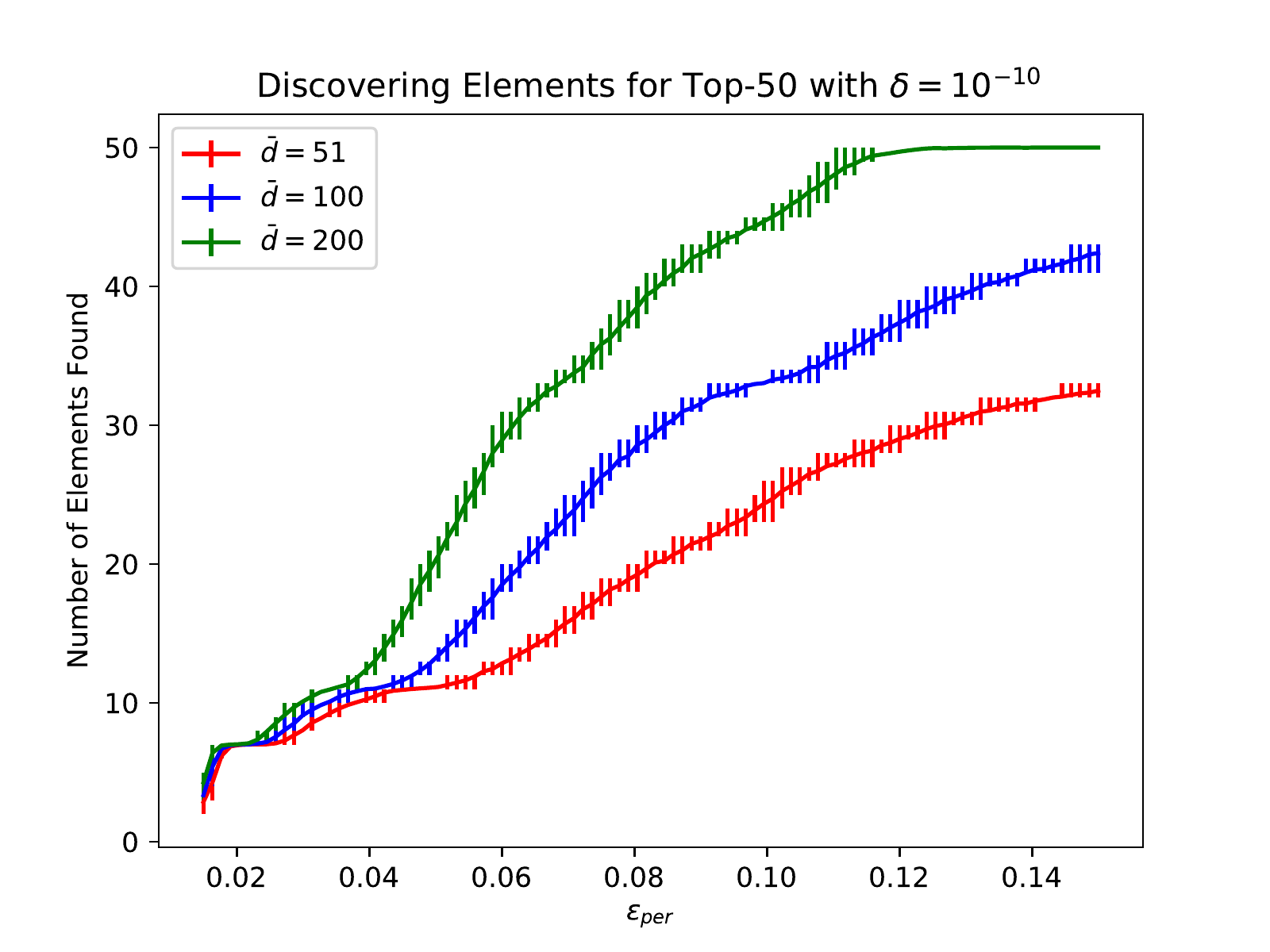}
\caption{The number of returned elements in $\rT{50,\bar{d},1}$ for a top-$50$ query with various $\bar{d}$.  We give the empirical average in 1000 trials and the (25\%,75\%) percentiles.  \label{fig:AlgoUnkGum}}
\end{figure}

We also empirically evaluate procedure $\rTE{\Delta,\bar{d},1}$ from Algorithm~\ref{algo:genLimitDomLap} in the unknown domain, $\Delta$-restricted sensitivity setting.  In Figure~\ref{fig:AlgoUnkLap} we show both the proportion of times in 1000 trials that each element was returned (right vertical axis) as well as the comparison between the noisy counts (in green) and the true counts (in red) that are returned for the discovered elements for a single trial (left vertical axis).  In each plot there is a privacy parameter $\diffpper \in \{0.1, 0.2\}$, with fixed $\delta = 10^{-10}$.  We ask for the top primary job titles of members that engaged with articles about \emph{privacy} or \emph{California}.  We assume that any one member cannot have more than one primary job title, hence $\Delta = 1$, and fetch $\bar{d} = 1000$ results from Pinot.

For the budget manager, we have a fixed budget for each marketing partner.  Once the privacy budget is depleted, a marketing partner would recycle old queries to get the same results or wait some fixed amount of time for the privacy budget to be refreshed. This policy decision for the rate in which to refresh the budget is dependent on how often the underlying dataset gets renewed and the characteristics of the underlying dataset.  In order to maintain consistency across the same queries on the same dataset, we use the same seed in the pseudorandom noise, as in \cite{KenthapadiTr18}.  

\begin{figure}
\centering
\includegraphics[width=.47\textwidth]{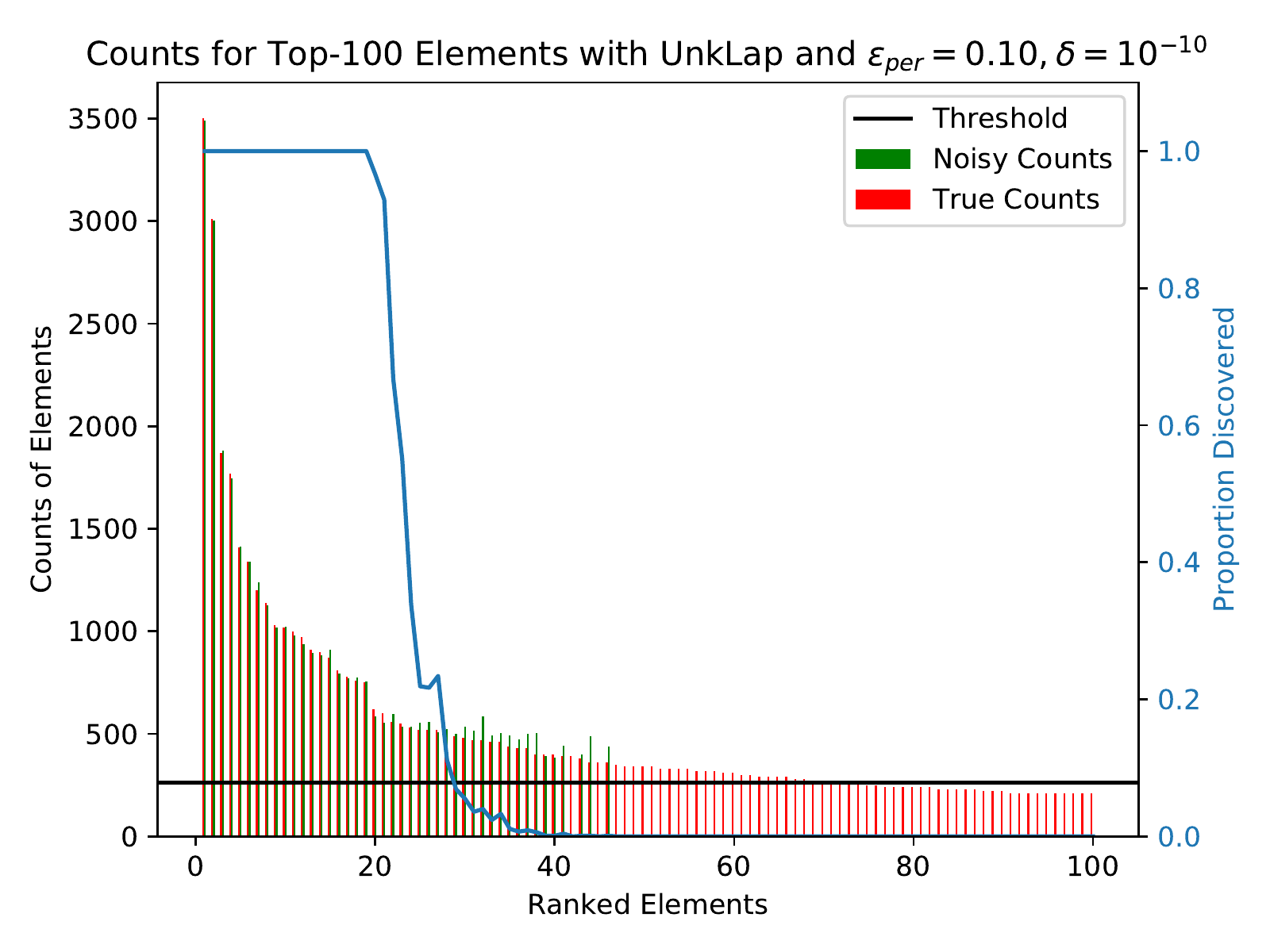}
\includegraphics[width=.47\textwidth]{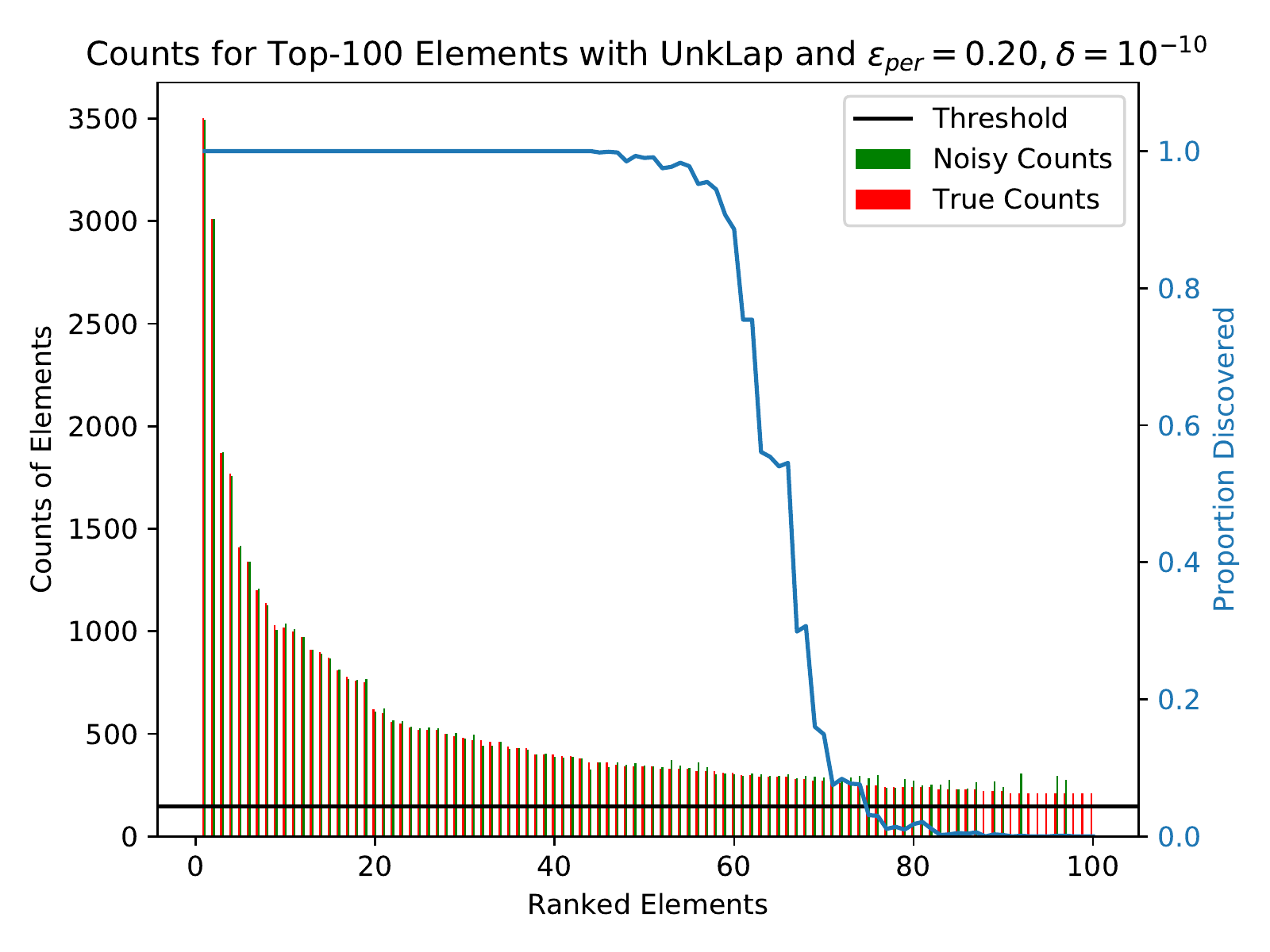}
\caption{The noisy counts (left $y$-axis) of the discovered elements returned in $\rTE{1,1000,1}$ for a top-$100$ query as well as the proportion (right $y$-axis) in which various elements in $1000$ independent trials were discovered.  The top plot gives results for $\diffpper = 0.1$ and the bottom plot gives $\diffpper = 0.2$. \label{fig:AlgoUnkLap}}
\end{figure}


\section{Deployment Considerations \label{sect:deployment}}

We now discuss our approach in deploying such a system that integrated multiple components, including Pinot for data analytics, differentially private algorithms, and a privacy budget management system.  Not knowing how external marketing partners would respond to budgeting access to queries and noisy results, we proceeded with a phased approach deploying our privacy system, first by turning on our privatized algorithms and only tracking usage of budget and then moving to enforce a given privacy budget. Recall that there are multiple parameters to set in our system and we detail the approach that we took to set them.   Ultimately, our privacy approach was guided by developing differentially private algorithms so that privacy loss could be quantified and we focused on specific attacks for how to set parameters.  Currently, the Audience Engagement API is still only available to a set of trusted users and pre-general availability (pre-GA). 

\subsection{Phased Approach of Deploying our Privacy System}

Given the multiple teams and components that made up our privacy system, we wanted to better understand the impact to the customer when the various components were enabled.  The main questions we faced included: how would external partners react with getting fewer results than they asked for (due to our private algorithms setting a data dependent threshold), and how much budget should we set without drastically modifying the behavior of how the external marketing partners interacted with the API.  We then sought to answer each question separately, with a phased approach of turning on each component.  The Audience Engagement API was planned to go through a soft launch period, followed by onboarding trusted partners, to then full scale deployment.  This allowed us to use the stages to deploy the different components of our privacy system.  

We first deployed our differentially private algorithms and identified the various columns in the data table that had a known/unknown domain or a restricted/unrestricted sensitivity.  We then tracked the privacy budget usage of the analysts that queried the API.  Recall that we refer to the quantities $k^*$ and $\ell^*$ as \emph{information budget} and \emph{call budget}, respectively.  Figure~\ref{fig:budgetUsage} shows the percentage of users that would exhaust their information and call budgets over the number of days since their budget was refreshed with various cut off amounts.  In this application, we refresh the budget once a month (see next subsection for more details on this), so we show the number of days since refresh on the $x$-axis because not all users will make a query on the first of the month and their budget does not refresh until they make their first query in the month.  Further, we see that there is a percentage of users that would exhaust their budget on the first day their budget is refreshed due to asking a batch of queries at once.

\begin{figure}[h]
\centering
\includegraphics[width=0.47\textwidth]{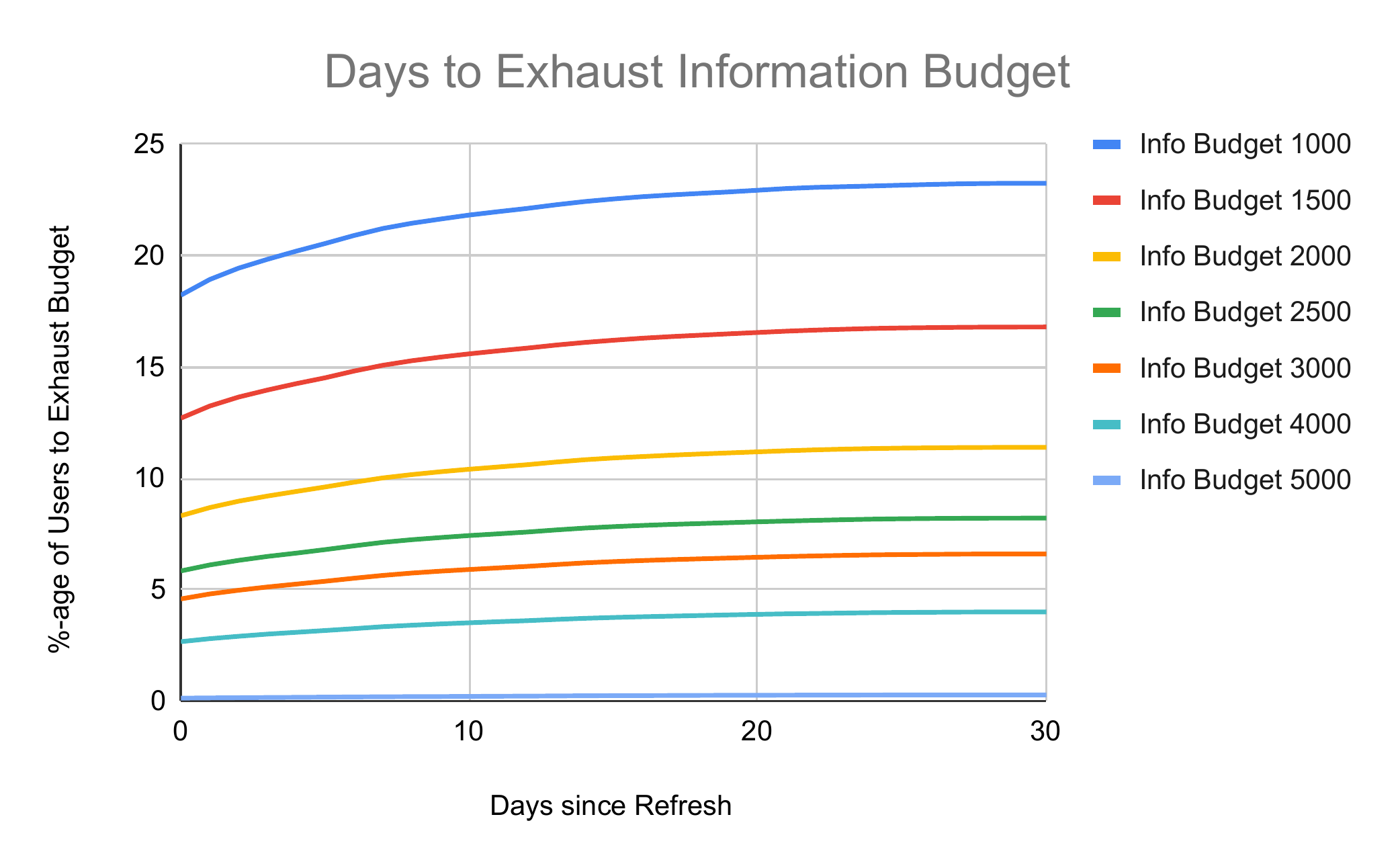}
\includegraphics[width=0.47\textwidth]{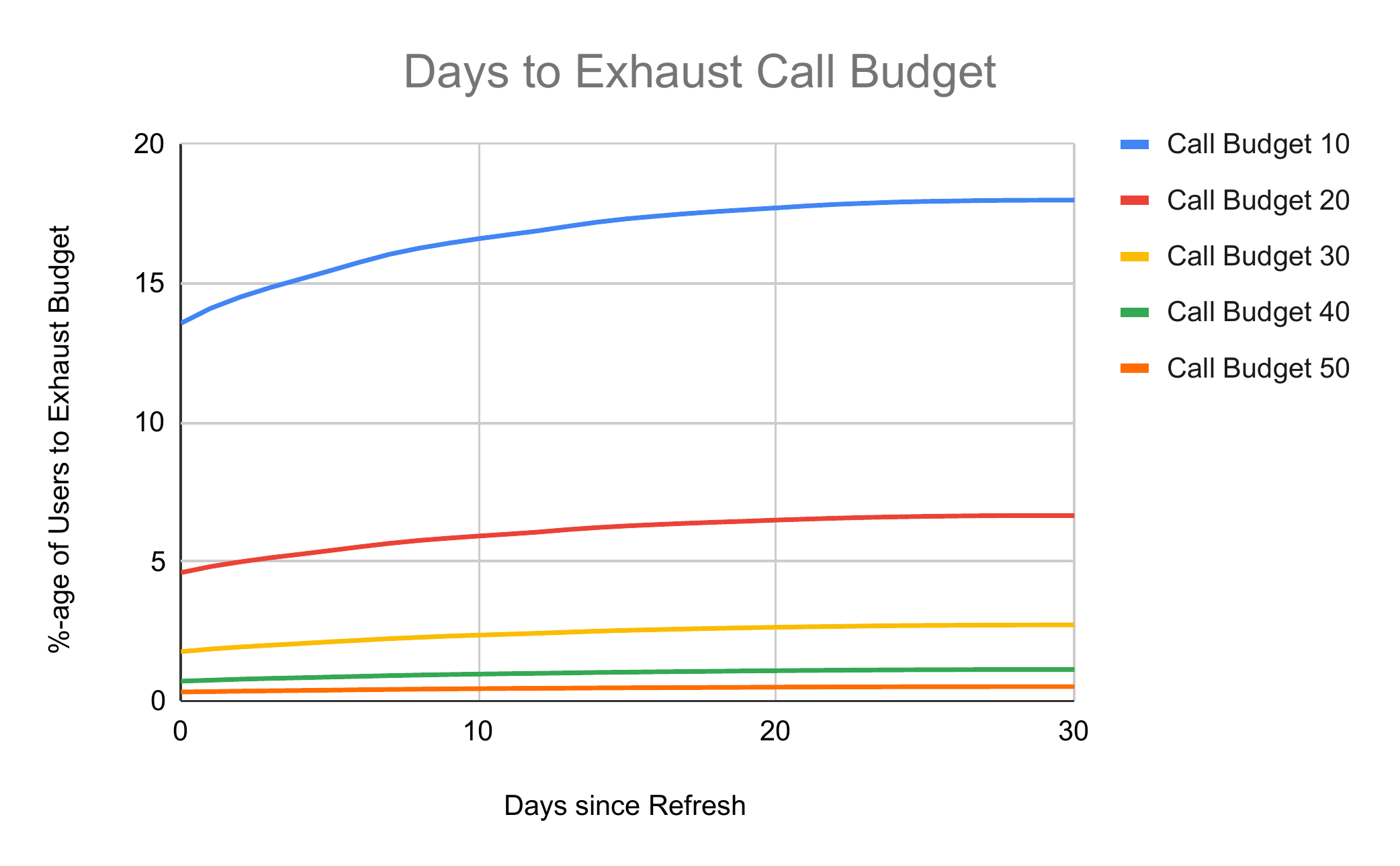}
\caption{We plot the percentage of analysts that exceed their information or call budgets from the time that budget is refreshed for various information and call budgets. \label{fig:budgetUsage}}
\end{figure}

The main takeaway from the plots in Figure~\ref{fig:budgetUsage} is that we can set information and call budgets in a way that does not limit a vast majority of users.  In particular, setting information budget to 3000 would not impact more than 93\% of users and a call budget of 30 would not impact more than 95\% of users.  

\subsection{Consistency and Data Refresh}

It is important to point out, from a product utility perspective, that data consistency is crucial.  Although this might seem to contradict the inherent randomness of differential privacy, we still want to ensure that if someone asks the same query, then they get the same result.  To ensure consistent results, we use the pseudorandom seed generation from \cite{KenthapadiTr18} for our randomized algorithms.  Hence, with the same query, we will use the same pseudorandom seed and hence the results will be the same, unless the underlying dataset has changed.  
The trusted parties who were granted access to the API all built UIs to facilitate access to their users, which prevents arbitrary query construction with small syntactic changes that leaves the semantics of the query unchanged.

This pseudorandom seed has an extra benefit for privacy as well, since if an analyst asks the same query multiple times, the random answers will not concentrate to the true answer.  Ensuring that private answers do not concentrate to the true result is one of the primary reasons for the privacy budget management, but setting information cost and call cost to ensure that the exact same query ran repeatedly does not concentrate to the true value would lead to overly pessimistic budgets to the point that the privacy system is not usable. 

The data for Audience Engagement is being placed into Pinot on a daily basis and retained for 30 days.  Hence, each day the data can potentially change, but not within the same day.  We then use the query and the date to determine the pseudorandom seed, otherwise if we only use the query, then the noise added to a specific query would always be the same, despite the data changing.  Note that we further use a secure key in the pseudorandom seed so that one cannot determine the seed only from the query and the date.

Recall that we are assuming analysts do not collude with each other and do not share the privatized results they receive, since this would mean essentially multiplying the information and call budgets that we enforced.   However, with the pseudorandom seed being generated the same way for all analysts, we know that each analyst is receiving the same result so that even if they colluded, they could not average their results to get more confidence in the true result.  

\subsection{Rationale for Parameters in our Privacy System}

There are ultimately four different parameters to set in our described privacy system: information budget $k^*$, call budget $\ell^*$, $\diffpper$, and $\delta$.  Note that with distinct counts, the $\ell_\infty$-sensitivity is 1 for all queries and for any query that is deemed restricted sensitivity, we set $\Delta = 1$.  The $k$ parameter for top-$k$ is an input from the analyst in each query and for unknown domain queries, $\bar{d}$ is set to be as large as possible without harming the efficiency requirements of the product, typically $\bar{d} = \max\{10k, 1000 \}$.

Part of the appeal of differential privacy is that it provides a worst case guarantee against privacy attacks, and can be simply stated as preventing an adversary from distinguishing whether a target's data was used in the analysis or not.  This strong protection stops being very meaningful once the privacy loss parameter $\diffp$ becomes large, say even larger than 1.  However, deployments of differential privacy have quoted much larger privacy parameters than 1, see for example \cite{ApplePrivacy17}, and more recent works in private ML have used much larger parameters \cite{Abadietal16}.  Further, these quoted parameters are for a one time calculation, rather than over multiple queries.  Only deploying privacy systems that incorporate differential privacy with small $\epsilon$ would limit its applicability.  In particular, in our setting we might only be able to allow for a single top-10 result before an analyst has exhausted his or her budget.  Boiling down the entire privacy considerations of a complicated system to whether a couple of parameters stay below some arbitrary privacy threshold for all deployments seems overly simplistic.  Other privacy safeguards can be added to increase the overall privacy, such as subsampling the dataset, as is done in this application. Further, the parameters in our system allow us to easily improve the overall privacy gains by modifying parameters.  Providing the tuning knob between 100\% utility and 100\% privacy is incredibly helpful in showing the impact that differential privacy has on the overall product.

The question then is, what protections can differential privacy provide, even with large $\epsilon$ parameters.  For this, we consider several different attacks, each used to set a certain parameter.  This is not meant as an exhaustive list of all the attacks we considered, nor does it mean that these certain attacks are expected.  This is merely to provide additional context to how privacy parameters can be set and might be useful for other privacy practitioners to use.

\subsubsection{Determining $\diffpper$}

We consider the scenario where the dataset remains the same over the course of the data retention period (30 days) and each day an analyst asks the same query on that dataset.  Recall how we set a pseudorandom seed for the same query and that it changes each day.  Thus, the analyst would get 30 different noisy results on the same count.  We then want to know the probability that the average of these noisy values will be within a tolerance of the true value.  We set this tolerance to $1/2$, since this would mean that the analyst rounding to the nearest integer would reveal the true count.  We then want to determine the following probability where $X_1, \cdots X_{30} \stackrel{i.i.d.}{\sim} \lap(2/\diffpper)$,
\[
\Pr\left[\left| \tfrac{1}{30} \sum_{i=1}^{30} X_i \right| < 1/2\right].
\]
We approximate this probability with a Normal distribution, so that we have
\[
\Pr\left[\left| \tfrac{1}{30} \sum_{i=1}^{30} X_i \right| < 1/2\right] \approx \Pr\left[ |\N(0,1)| \leq \frac{\diffpper \sqrt{30}}{4 \sqrt{2}} \right].
\]
Our aim is to reduce the chance of this attack, while also not adding too much noise to each count.  We then sought to ensure roughly a 90\% chance of no such attack.  Plugging in $\diffpper = 0.15$ leads to about an 11\% probability of this attack being successful.  Also recall that this attack will only be successful if the data does not change for the query over 30 days.  

\subsubsection{Determining Information Budget}

One of the primary reasons for exploring differential privacy for this use case was differencing attacks.  Consider the setting where an analyst asks two queries where they know that there is a single person different between the two in the unrestricted sensitivity setting.  Despite the noise that we add to each count in each query, it is possible that the noise is small for some elements so it is clear what the true count was before noise.  This happens when the noise is smaller than $1/2$ for a single count.  Hence we want to compute the following probability, which can be written in terms of an exponential random variable $\texttt{Exp}$,
\[
p \defeq \Pr[|\lap(2/\diffpper) | < 1/2].  = \Pr[\texttt{Exp}(\diffpper/2) < 1/2].
\]

For a top-$k$ query, we can expect to see $p k$ (assume integer valued) many elements that have noisy count within half of the true count.  If an adversary were to do a differencing attack, she would ask another top-$k$ query and there would be fresh noise added.  Hence, there would again be an expected $p k$ many noisy counts that are close to the true counts.  The adversary does not know which elements in both queries have noisy counts within $1/2$ of the true count, but we want to make sure these sets of elements do not overlap.  If these elements with small noise do overlap in the two top-$k$ results, then a difference between the two results might show the actual difference between the two.  

We now want to prevent the possibility of these small count elements to overlap in the two top-$k$ queries. Let’s fix the set of $pk$ elements that had noisy counts within half of the true counts in the first top-$k$.  In the second top-$k$, we know that again there are expected to be  $pk$ elements, but where they are is random.  Hence we get a uniformly random set of $pk$ elements in the second top-$k$ result and want to know what is the expected size of the intersection between this set of $pk$ elements and the $pk$ elements from the first top-$k$.  The probability that this intersection is of size $s$ is the following:
\[
\frac{{ pk \choose s} {k - pk \choose pk - s}}{ {k \choose pk} }.
\]
This is a hypergeometric distribution and has expectation $p^2k$.  To reduce the chance that these two sets of size $pk$ overlap, we set the expected size of the intersection to be less than $1$, i.e. $k<1/p^2$.  

Recall from the previous attack that we have $\diffpper = 0.15$, we get $p=0.0368$ and we then can use $k=738$.  Our information cost is the total number of results that can be returned plus one for the optimized threshold calculation in the unknown domain setting, which would bound the information cost by $2k+2$, from these two top-$k$ results.  Note that we also return counts for the elements that we find, which also increases the total information cost by $2k$.  Hence, we can set information budget $k^* = 4 \cdot 738 +2 \approx 3000$, which from Figure~\ref{fig:budgetUsage}, we see that more than 93\% of analysts would not be impacted.

\subsubsection{Determining $\delta$ and Call Budget}

Our unknown domain algorithms include a threshold so that elements with a single user contribution (unique count) should not be shown in any result.  However, noise is added to the threshold to ensure differential privacy, so we want to be able to control the chance that a unique count with noise becomes larger than the noisy threshold.  Hence, we want to bound the probability that this can occur for a single count and then take a union bound over all $\bar{d}$ counts that can be returned from Pinot.  For the $\rT{}$ algorithm, we will write $Z_1, Z_2 \sim \gum(1/\diffpper)$ and $h_{(i)}$ as the $i$th ranked count in the input histogram $\bbh$.  We then consider the following probability of a \emph{bad single event} $B_i$ where $i \leq \bar{d}$
\[
\Pr[B_i] \defeq \Pr[h_{(i)} + Z_1 > h_{(\bar{d} + 1)} + 1 + \log(\bar{d}/\delta)/\diffpper + Z_2].
\]
The worst case scenario is where every element in the histogram has count equal to 1, meaning only one user contributed to the counts and so each count is a unique count.  Note that in $\rT{}$, the threshold actually uses $\log(\bar{k}/\delta)$ where there is an additional step to optimize from $\bar{d}$ to a smaller $\bar{k}$, but here we use the larger $\bar{d}$ to be more pessimistic.  
Noting that the difference $Z_1 - Z_2$ is distributed as a logistic random variable $\texttt{Log}$, we have the following
\begin{align*}
\Pr[B_i] & \leq \Pr\left[\texttt{Log}(1/\diffpper) > 1 + \frac{\log(\bar{d}/\delta)}{\diffpper}\right] \\
& = 1 - \frac{1}{1 + e^{-\diffpper} \delta / \bar{d}}. 
\end{align*}
We then want to bound the event that any of the $\bar{d}$ elements can appear above the threshold, hence 
\begin{align*}
\Pr[\cup_{i=1}^{\bar{d} } B_i] & \leq \sum_{i=1}^{\bar{d}}  \Pr[B_i] \\ 
& \leq \bar{d} \cdot \left(1 - \frac{1}{1 + e^{-\diffpper} \delta / \bar{d}} \right) \leq \delta e^{-\diffpper}
\end{align*}
Now, we want to make sure that the chance of this occurring over all queries is small.  We use $\ell^*$ to denote the call cost which is how many top-$k$ queries in the unknown domain setting are used and are the only types of queries that risk showing results that a single user contributed.  Hence we will write $B_i^\ell$ to be a bad event for count $i \leq \bar{d}_\ell$ in the $\ell$th query and take a union bound over all $\ellmax$ such queries.
\[
\Pr[\cup_{\ell=1}^{\ellmax} \cup_{i=1}^{\bar{d}_\ell} B_i^\ell] \leq  \ell^* \delta e^{-\diffpper}.
\]

We aim for a 1 in a hundred million chance and then also reference the cost budget usage in Figure~\ref{fig:budgetUsage}.  Setting $\ell^* = 30$ would not impact more than 95\% of analysts and using $\delta = 10^{-10}$ with $\diffpper = 0.15$ gives the overall probability bound of close to 1 in 400 million.

\begin{table}[H]
\makeatletter
\newcommand\footnoteref[1]{\protected@xdef\@thefnmark{\ref{#1}}\@footnotemark}
\makeatother
\begin{minipage}{\textwidth}
\begin{tabular}{|c|c|c|c|c|}
\rowcolor[HTML]{C0C0C0} 
\textbf{Use Case} & \textbf{\begin{tabular}[c]{@{}c@{}}Privacy\\ Model  \end{tabular}} & \textbf{\begin{tabular}[c]{@{}c@{}}DP Algorithm\\ Parameters \\ $(\epsilon,\delta) $ \end{tabular}} & \textbf{\begin{tabular}[c]{@{}c@{}}Daily DP \\ Parameters \\  $(\epsilon_\text{day},\delta_\text{day})$\end{tabular}} & \textbf{\begin{tabular}[c]{@{}c@{}} Monthly DP\\ Parameters\\ $(\epsilon_\text{month},\delta_\text{month}) $\end{tabular}} \\
\hline
\begin{tabular}[c]{@{}c@{}}Google - RAPPOR \citep{ErlingssonPiKo14}\\ Chrome Homepages\end{tabular}               
& Local\footnoteref{Memoization}   
& $(0.534, 0)$                                           
& \begin{tabular}[c]{@{}c@{}}$(25.63,0)$ \\ 30 min reporting \end{tabular}                      
& $ (769, 0)$                  
\\
\hline
\begin{tabular}[c]{@{}c@{}}Apple - Safari\\ Domains \citep{ApplePrivacy17}  \end{tabular}                                                                 
& Local          
& $(4, 0)$                                                            
& $(8, 0)$\footnoteref{ApplePage}                                                       
& $(240, 0) $                                                                        
\\
\hline
Apple - Emojis   \citep{ApplePrivacy17}                                                              
& Local            
& $(4, 0)$                                                            
& $(4,0)$\footnote{\label{ApplePage}\url{https://www.apple.com/privacy/docs/Differential_Privacy_Overview.pdf}}                                                       
& $(120,0) $                                                                    
\\
\hline
\begin{tabular}[c]{@{}c@{}}Microsoft - Telemetry\\ Collection per App \citep{DingKuYe17}\end{tabular} 
& Local\footnote{\label{Memoization}Memoization ensures that repeated records have much smaller overall privacy loss.  This table shows the overall privacy loss for users that generate distinct records.}  
& $(0.686,0) $                                                                           
&  \begin{tabular}[c]{@{}c@{}}$(2.74,0)$ \\ 6 hour reporting\end{tabular}                                                                                 
&  \begin{tabular}[c]{@{}c@{}}$(82.2,0)$ \end{tabular}                                               
\\
\hline
\begin{tabular}[c]{@{}c@{}}Google - Mobility\\ Reports \citep{GoogleMobility} \end{tabular}                                                     
& Global           
& \begin{tabular}[c]{@{}c@{}}$(0.11,0)$ or \\ $(0.22,0)$ \end{tabular}                                                                                     
& $(2.64,0)$\footnote{We add up the parameters from \citep{GoogleMobility} in daily visits in public places $(\epsilon = 1.74)$, residential $(\epsilon = 0.44)$, and workplaces $(\epsilon = 0.44)$.}                      
& $(79.2,0)$                                      
\\
\hline
\begin{tabular}[c]{@{}c@{}}Microsoft - Assistive \\ AI\footnote{\url{https://www.microsoft.com/en-us/research/group/msai/articles/assistive-ai-makes-replying-easier-2}} \end{tabular}                                                        
& Global           
& $(4,10^{-7})$                                                              
&  Not available                                                       
&  Not available                                                                      
\\
\hline
\begin{tabular}[c]{@{}c@{}}LinkedIn - Audience \\ Engagement API\footnote{ pre-GA status}\end{tabular}      
& Global           
& $(0.15,10^{-10})$                                                                
&  ---                                                
& \begin{tabular}[c]{@{}c@{}} $(34.9,7\times10^{-9})$ \end{tabular}                            
\\
\hline
\end{tabular}
\caption{Privacy parameters for existing deployments of privacy systems that use differentially private algorithms.  Note that some parameters can be improved with a slightly larger $\delta_\text{month}$.\label{table:epsilons}}
\end{minipage}
\end{table}

\subsection{Overall Privacy Guarantee}

With these proposed parameter values for $\diffpper, \delta, \kmax,\ellmax$ and applying Theorem~\ref{thm:final_parameters}, we get a final $(34.9, 7 \times 10^{-9})$-DP monthly guarantee.  We want to put this guarantee in the context of other deployments of privacy systems that use differential privacy.  Note that some privacy systems that adopted differential privacy are in the local privacy model, which is a more stringent privacy setting than the global model that we consider here.\footnote{It is possible to compute the global DP parameters when local DP is used, see \cite{ErlingssonFeMiRaTaTh19}, \cite{BalleBeGaNi19}, and \cite{CheuSmUlZeZh19}}  Further, we used parameters that were publicly released, so they may differ from current deployments.  

In Table~\ref{table:epsilons}, we identify different use cases that have adopted differential privacy and have released their privacy parameters along with how often data and reports are refreshed, thus allowing us to compute daily and monthly DP parameters.  We focused on deployments where data is continually updated in both local and global models of privacy.  It is important to point out that each privacy system includes additional safeguards beyond differentially private algorithms.  Some of these differences include subsampling users, permuting records, and the use of memoization, as in Google's RAPPOR \citep{ErlingssonPiKo14} and Microsoft's telemetry collection \citep{DingKuYe17}, to prevent longitudinal attacks when the same record is privatized with fresh noise repeatedly.  What we account for in the table is if a user's data changes in the local model then fresh noise would be added to each result and hence the privacy loss accumulates.


\section{Conclusion}
We have presented a privacy system that incorporates state of the art algorithms for releasing histograms and top-$k$ results in a differentially private way.  Also, we have shown how we track the privacy budget for multiple analysts that can query our API.  Combining the budget management service with DP algorithms allows us to make strong privacy guarantees of the overall system for any external partner that is allowed to make multiple, adaptively selected queries.  This privacy system allows us to track the amount of information that is being released to external partners via the API in a precise way so that we can make informed decisions in how we can balance privacy safeguards with the usefulness of the product.  We hope that this work demonstrates the feasibility of providing rigorous DP guarantees in systems that can scale.


\paragraph{Acknowledgements}
We would like to thank Adrian Cardoso, Mark Cesar, Stephen Lynch, Sofus Macskassy, Koray Mancuhan, Sajjad Moradi, Sergey Yekhanin, and the entire LinkedIn Data Science Applied Research team for their helpful feedback on this work.  Further, we thank Igor Perisic and Ya Xu for their support throughout this project.  

\clearpage 

\bibliographystyle{IEEEtran}
\bibliography{bib}

\clearpage

\appendix

\section{Omitted Analysis for Section~\ref{sect:restricted_sensitivity} \label{app:LapMax}}

We now go through the analysis for Algorithm~\ref{algo:genLimitDomLap}, in particular the proof of Lemma~\ref{lem:DDR19}.  The differences between Algorithm~\ref{algo:genLimitDomLap} and the version that appeared as Algorithm 4 in \cite{DurfeeRo19} is that we are returning counts as well as indices, we do not limit the number of outcomes to be at most $k$ (since it is not a parameter), and we allow for counts to increase or decrease by $\tau\geq 1$ in neighboring datasets.  As we will mainly be borrowing the analysis in \cite{DurfeeRo19} we will change $\bar{d}$ to $\bk$ to better match the statements in that work. We then introduce the following algorithm, which we will show has the same distribution as $\rTE{\Delta,\bk,\tau}(\bbh)$.

\begin{definition}[Limited Histogram Report Noisy Counts]\label{defn:rnmk}
We assume that the $\ell_\infty$ sensitivity between any neighboring histograms is $\tau$.  We define the limited histogram report noisy counts to be $\rnm{\bk,\tau}$ that takes as input a histogram along with a domain set of indices and returns an ordered list of counts with the corresponding index, where
$
\rnm{\bk,\tau}(\bbh,\cD) =
\left( \{v_{(1)}, i_{(1)}\},...,\{ v_{(\bot)} ,\bot \}\right) 
$
and $(v_{(1)},...,v_{(\bot)})$ is the sorted list of $v_i = h_{(i)} + \lap(2\tau\Delta/\diffpper)$ for each $i \in \cD$ and $v_\bot = h_{(\bk + 1)} + \tau \left( 1 + 2\Delta\log(\Delta/\hat\delta)/\diffpper \right) + \lap(2 \tau\Delta/\diffpper)$ with $\hat\delta$ given in Algorithm~\ref{algo:genLimitDomLap} as a function of $\delta >0$.
\end{definition}

We have the following that connects $\rnm{\bk,\tau}$ with $\rTE{\Delta,\bk,\tau}$.
\begin{lemma}\label{lem:lap_equiv}
For any histogram $\bbh$, we have that both mechanisms $\rnm{\bk,\tau}(\bbh,\domain{\bk}{\bbh})$ and $\rTE{\bk,\tau}(\bbh)$ produce outcomes that are equal in distribution.
\end{lemma}

If we fix a domain $\cD$ beforehand, then we have the following privacy statement.  Note that the privacy of $\rnm{\bk,\tau}(\bbh, \cD)$ follows from the Laplace mechanism \cite{DworkMcNiSm06} being $\diffpper$-DP.  This is what allows us to output the counts as well as the indices.  We just need to ensure that $i_{(\bk + 1)} \notin \cD$ because then if it was, then changing one index would change the count of both $h_{(\bk + 1)}$ and $h_{\bot}=h_{(\bk + 1)} + \tau \left( 1 +  \Delta\log(\Delta/\delta)/\diffpper \right)$.
\begin{lemma}\label{lem:lap_restricted}
For any fixed $\cD \subseteq [d]$ and neighbors $\bbh,\bbh'$ such that $i_{(\bk + 1)},i'_{(\bk + 1)} \notin \cD$, then for any set of outcomes $T$,

 \begin{align*}
 & \Pr[\rnm{\bk,\tau}(\bbh, \cD) \in T] \\
 & \qquad \leq e^{\diffpper/2}\Pr[\rnm{\bk,\tau}(\bbh',\cD) \in T].
\end{align*}
\end{lemma}

As was done in Durfee and Rogers \cite{DurfeeRo19}, we can carefully account for the \emph{good}  (can bound the privacy loss) and \emph{bad} (can bound these events with small probability) sets. 
Note that the outcome set of $\rTE{\Delta,\bk,\tau}$ is a superset of Algorithm 4 in \cite{DurfeeRo19} when $k = \bk$, and it is straightforward to see that these algorithms have the same distribution with respect to index output (ignoring the counts output from $\rTE{\Delta,\bk,\tau}$). 
Therefore, all the bounds on the \emph{bad} outcomes will still hold for our setting, and the analysis then follows from results in Section 6 of \cite{DurfeeRo19}, where we state each result here.  

\begin{definition}
Given two neighboring histograms $\bbh,\bbh'$, we define $\cS_{\lap}$ as the outcome set of $\rTE{\Delta,\bk,\tau}(\bbh,\domain{\bk}{\bbh})$ (both indices and counts) and the outcome set of $\rTE{\Delta,\bk,\tau}(\bbh',\domain{\bk}{\bbh'})$ as $\cS'_{\lap}$.

We then define the \emph{bad outcomes} as 
$
\cS^\delta_{\lap} \defeq \cS_{\lap} \setminus \cS'_{\lap}
$
and 
$ \cS'^\delta_{\lap} \defeq \cS'_{\lap} \setminus \cS_{\lap}.
$
\label{defn:eps_delt_setsLAP}
\end{definition}

\begin{lemma}\label{lem:lap_del}
For $\Delta$-restricted sensitivity neighbors $\bbh,\bbh'$, we have

\begin{align}
& \Pr[\rnm{\bk,\tau}(\bbh,\domainE{\bk}{\bbh}) \in \cS^\delta_{\lap}] \nonumber \\
& \qquad \leq  \delta / 4\cdot \left( 3 + \ln(\Delta/\delta) \right) \eqdef \bar{\delta}
\label{eq:barDelta}
\end{align}
\end{lemma}

\begin{lemma}\label{lem:lap_eps}
For any neighboring histograms $\bbh,\bbh'$ and for any $S \subseteq \cS_{\lap} \cap \cS'_{\lap}$, we let $\cD^\diffpper = \domainE{\bk}{\bbh} \cap \domainE{\bk}{\bbh'}$ and we must have the following for $\bar{\delta}$ given in \eqref{eq:barDelta}
\begin{align*}
& \Pr[\rnm{\bk,\tau}(\bbh, \domainE{\bk}{\bbh}) \in S] \\
& \qquad  \leq \Pr[\rnm{\bk,\tau}(\bbh,\cD^\diffpper) \in S] \\
&\qquad \leq \Pr[\rnm{\bk,\tau}(\bbh,\domainE{\bk}{\bbh}) \in S] + \bar{\delta}
\end{align*}
\end{lemma}

\begin{lemma}\label{lem:eps_delta_lap}
For any neighboring histograms $\bbh,\bbh'$ and any $S \subseteq \cS_{\lap}$, then for $\bar{\delta}$ given in \eqref{eq:barDelta},
\begin{align*}
\Pr& [\rnm{\bk,\tau}(\bbh,\domainE{\bk}{\bbh}) \in S] \\ 
& \qquad \leq e^{\diffpper/2} \Pr[\rnm{\bk,\tau}(\bbh',\domainE{\bk}{\bbh'}) \in S] \\
&\qquad \qquad \qquad + (e^{\diffpper/2} + 1)\bar{\delta}.
\end{align*}
\end{lemma}
\begin{proof}
We use the above results to get the following inequalities.
\begin{align*}
\Pr & [\rnm{\bk,\tau}(\bbh,\domainE{\bk}{\bbh}) \in S] \\
& =  \Pr[\rnm{\bk,\tau}(\bbh,\domainE{\bk}{\bbh}) \in S \cap \{ \cS_{\lap} \cap \cS_{\lap}' \}]  \\
& \qquad +  \Pr[\rnm{\bk,\tau}(\bbh,\domainE{\bk}{\bbh}) \in S \cap \{ \cS_{\lap}^\delta \}]\\
&  \leq \Pr[\rnm{\bk,\tau}(\bbh,\cD^\diffpper) \in S\cap \{ \cS_{\lap} \cap \cS_{\lap}' \}] \\
& \qquad + \bar{\delta} \\
& \leq e^{\diffpper/2} \\
& \qquad \cdot \Pr[\rnm{\bk,\tau}(\bbh',\cD^\diffpper) \in S\cap \{ \cS_{\lap} \cap \cS_{\lap}' \}]  + \bar{\delta} \\
& \leq e^{\diffpper/2} \\
& \quad \left( \Pr[\rnm{\bk,\tau}(\bbh',\domainE{\bk}{\bbh'}) \in S\cap \{ \cS_{\lap} \cap \cS_{\lap}' \}]  + \bar{\delta} \right) \\
& \qquad + \bar{\delta} \\
& \leq  e^{\diffpper/2} \Pr[\rnm{\bk,\tau}(\bbh',\domainE{\bk}{\bbh'}) \in S] + (e^{\diffpper/2} + 1)\bar{\delta}.
\end{align*}
\end{proof}

\end{document}